\documentclass{article}
\usepackage[margin=1in]{geometry}
\usepackage[utf8]{inputenc}

\usepackage{amsmath}
\usepackage{amsfonts}
\usepackage{amssymb}
\usepackage{amsthm}
\usepackage{mathtools}
\usepackage{graphicx}
\usepackage{float}
\usepackage{tabularx}
\usepackage{color}
\usepackage{transparent}
\usepackage{algorithm}
\usepackage{algorithmic}
\usepackage{subcaption}

\renewcommand{\vec}{\mathbf}

\renewcommand{\bar}{\overline}
\renewcommand{\hat}{\widehat}

\DeclareMathOperator*{\argminbelow}{arg\,min}

\newcommand{\ubar}{\underline}
\newcommand{\bbR}{\mathbb{R}}

\newcommand{\calU}{\mathcal{U}}
\newcommand{\calK}{\mathcal{K}}

\newcommand{\indenti}[1]{\transparent{0}{#1}\transparent{1}}

\newtheorem{theorem}{Theorem}

\newtheorem{lemma}{Lemma}

\newtheorem{corollary}{Corollary}

\newtheorem{definition}{Definition}
\newtheorem{remark}{Remark}
\newtheorem{conjecture}{Conjecture}

\newtheorem{problem}{Problem}

\title{Analytical Construction of CBF-Based Safety Filters for Simultaneous State and Input Constraints (Extended Version)}
\author{Peter A. Fisher and Anuradha M. Annaswamy
\thanks{This work was supported by the Boeing Strategic University Initiative and by the Air Force Research Laboratory.}
\thanks{The authors are with the Department of Mechanical Engineering, Massachusetts Institute of Technology, Cambridge, MA, 02139 USA}
\thanks{Corresponding author: P.A. Fisher, \ {\tt pafisher@mit.edu}}
}
\date{}

\begin{document}

\maketitle

% \begin{abstract}
%     We revisit the problem explored in \cite{doeser2020invariant} of guaranteeing satisfaction of multiple simultaneous state constraints applied to a single-input, single-output plant consisting of a chain of $n$ integrators subject to input limitations. For this problem setting, we derive an analytic, easy-to-implement safety filter which respects input limitations and ensures forward-invariance of all state constraints simultaneously. Additionally, we provide a straightforward extension to the multi-input, multi-output chained integrator setting, and provide an analytic safety filter guaranteeing satisfaction of arbitrarily many simultaneous hyperplane constraints on the output vector. Whereas the approach in \cite{doeser2020invariant} obtains maximal invariant sets, our approach trades off some degree of conservatism in exchange for a recursive safety filter which is analytic for any arbitrary $n \geq 1$.
% \end{abstract}

\begin{abstract}
    We consider the problem of controlling a plant consisting of an $n$th-order chain of integrators under input saturation, with the goal of ensuring satisfaction of simultaneous upper and lower constraints on every state variable while accomplishing some nominal control objective as closely as possible. In contrast to prior work, we derive a filter on the control input for a single state constraint which is analytical for any $n \geq 1$, and develop a guaranteed method of tuning multiple simultaneous such filters to prevent conflicts and allow for simultaneous satisfaction of any number of state constraints. We also discuss an extension of our approach to the multi-input $n$th-order integrator setting.
\end{abstract}

\section{Introduction} \label{sec:introduction}

% With ever-increasing speed and fervor, the market forces of the 21st century seek to put autonomous systems on a collision course with humanity. From household helpers to delivery drones, self-driving cars, and urban air mobility, robots and autonomous systems have been and will continue to be integrated into our lives and surroundings. With increased proximity to humanity, however, comes an increased need for safety guarantees. Thus, there is a rapidly growing emphasis within the controls literature on methods for guaranteeing enforcement of state constraints.

% Within this body of literature, control barrier functions (CBFs) \cite{ames2019CBF} have gained popularity for several reasons. CBFs are intuitive and can be straightforward to implement for many problem settings. Additionally, CBFs have an important property that their alterations to a nominal control signal are on an ``as-needed" basis, which is to say, they constrain the control input only when necessary and are otherwise inactive. This is especially important for high-performance applications wherein the system needs to be able to act aggressively whenever doing so would not conflict with safety.

Control barrier functions (CBFs) \cite{ames2019CBF} have gained traction in recent years as a tool for guaranteeing enforcement of state constraints, which often arise due to safety considerations (for instance, ``don't hit that wall"). A state constraint can often be expressed as requiring the state vector to remain in a certain set for all time, and thus safety is in the CBF literature is synonymous with forward-invariance of a ``safe set." A myriad of potential applications have been considered, such as walking robots \cite{nguyen2016walkingCBFs}, self-driving cars \cite{gunter2022automatedvehicles,alan2023automatedvehicles}, and multi-agent systems \cite{pmlr-v229-zhang23h}.

Much of the earlier CBF literature considered problem settings where sufficient control authority was assumed, explicitly or implicitly, to exist for any input the safety filter might command.
% (see e.g. \cite{})
However, some of the more recent literature has considered safety in the presence of input constraints.
% However, all real actuators have bandwidth and saturation limits. Furthermore, the interactions between simultaneous state constraints may limit control authority.
% % For example, consider a multirotor drone linearized about hover. For this system, pitch angle is the second derivative of forward-and-back position. If a constraint is placed on the pitch angle - for instance, to ensure validity of the linearization - then the horizontal acceleration of the drone is constrained, limiting the drone's ability to brake when approaching a barrier at high speed.
One category of approach seeks to find a subset of the safe set which can be rendered invariant under input constraints using techniques such as Hamilton-Jacobi reachability and sum-of-squares programming  \cite{gurriet2018,mitchell2005,xu2018correctness}. Both of these techniques are computationally expensive and suffer from the curse of dimensionality. Another proposed approach is the use of a backup controller which is known to render a certain backup set forward-invariant under input limitations \cite{gurriet2020scalable,molnar2023safety}. The backup set approach is efficient, but can also be conservative and relies on MPC-style online optimization. Specialized methods of handling input constraints also exist for certain classes of dynamical systems, such as Euler-Lagrange systems \cite{cortez2020correct} or systems that evolve on a manifold \cite{wu2015safety}.

The concept of an input-constrained CBF (ICCBF) was introduced in \cite{agrawal2021ICCBFs} as a generalization of high-order CBFs (see \cite{xiao2019HOCBF}). The idea of an ICCBF is to iteratively select subsets of the original state constraints until one is left with a state constraint that is satisfiable under input limitations.
% It is unclear under what conditions this iterative process will converge for a general problem setting. Nevertheless, ICCBFs present a useful framework to build off of.
As the final solution allows one to accommodate state and input constraints simultaneously, ICCBFs present a useful framework; however, conditions under which the iterative process converges are somewhat unclear.

In \cite{doeser2020invariant}, the authors consider the problem of satisfying multiple state constraints simultaneously for a single-input chain of $n$ integrators under input saturation. Minimally-conservative analytical CBF filters are derived for the setting where $n \leq 4$ and are applied to control of a multirotor drone.
% However, the approach in \cite{doeser2020invariant} is restricted to $n \leq 4$ because it involves solving for the root of an $n$th-order polynomial.
In this paper, we revisit this problem and, taking inspiration from ICCBFs in \cite{agrawal2021ICCBFs}, derive analytic CBF filters for any $n \geq 1$, thus alleviating the requirement in \cite{doeser2020invariant} that $n \leq 4$. We further discuss extensions of our approach to the multi-input setting and demonstrate its efficacy in simulation.

The problem of satisfying multiple state constraints simultaneously under input saturation was also studied in \cite{breeden2023compositions}, and our approach is closely related. There are, however, a few key differences. The approach in \cite{breeden2023compositions} is able to consider a plant with unforced dynamics but is restricted to plants with relative degree 2 or less. Additionally, whenever there are conflicting CBFs (two or more CBFs orientated such that increasing the value of one requires decreasing the value of another), the approach in \cite{breeden2023compositions} requires an iterative process to narrow down the safe set, and the conditions under which this process converges are unclear. In contrast, while we restrict ourselves to multi-integrator plants, our approach is able to handle arbitrary relative degree, and in the single-input case, handles conflicting CBFs via a tuning algorithm which is proven to converge in a finite number of steps.

In summary, the main contributions of this paper are as follows:
\begin{enumerate}
    \item a safety filter for a single-input $n$th-order integrator plant which guarantees satisfaction of state constraints in the presents of input saturation and is analytical for any $n \geq 1$,
    \item a guaranteed method of tuning multiple such filters to guarantee satisfaction of multiple simultaneous state constraints, and
    \item a discussion of extending our approach to the multi-input $n$th-order integrator setting, validated in simulation.
\end{enumerate}
In Section \ref{sec:preliminaries}, we provide a few preliminaries and a useful technical lemma. Section \ref{sec:problem_statements} presents two problem statements for a single-input system, where Problem \ref{plm:simplified} considers a simplified setting with only one of the $n$ states constrained, while Problem \ref{plm:full} considers a general setting with all $n$ states constrained. These problem statements are addressed in Sections \ref{sec:simplified_solution} and \ref{sec:full_solution} respectively. Section \ref{sec:extension} discusses an extension of our approach from the single-input to the multi-input setting. Finally, Section \ref{sec:simulations} demonstrates our approach in simulation on a linearized quadrotor drone, and Section \ref{sec:conclusions_future_work} concludes. Proofs and supporting material can be found in the appendices.
\section{Preliminaries} \label{sec:preliminaries}

This section introduces definitions and background material for the remainder of the paper. For the remainder of this section, consider an autonomous dynamical system
\begin{equation} \label{eqn:prelim_plant}
    \dot{\vec{x}} = \vec{f}(\vec{x}) + G(\vec{x})\vec{u}
\end{equation}
where $\vec{x} \in \bbR^n$, $\vec{u} \in \bbR^m$, $\vec{f}(\vec{x})$ and $G(\vec{x})$ are locally Lipschitz, and $\vec{u}(t)$ is subject to the input constraint
\begin{equation} \label{eqn:input_constraint}
    \vec{u}(t) \in \calU \subset \bbR^m\ \forall t \geq 0.
\end{equation}
Consider also a ``safe set" $S \subset \bbR^n$ defined by a smooth function $h(\vec{x})$ as follows:
\begin{subequations} \label{eqn:safe_set}
    \begin{align}
        S &= \{\vec{x} \in \bbR^n : h(\vec{x}) \geq 0\} \label{eqn:S} \\
        \partial S &= \{\vec{x} \in \bbR^n : h(\vec{x}) = 0\} \label{eqn:S_boundary} \\
        \mathrm{int}(S) &= \{\vec{x} \in \bbR^n : h(\vec{x}) > 0\} \label{eqn:S_interior}
    \end{align}
\end{subequations}
where $\partial S$ and $\mathrm{int}(S)$ denote the boundary and interior of $S$ respectively.

\subsection{Definitions} \label{subsec:definitions}

The following definitions are widely used in the CBF literature:
\begin{definition}[\cite{ames2019CBF}]
    A continuous function $\alpha : [0, \infty) \to [0, \infty)$ is \textit{class-$\calK_\infty$} (denoted $\alpha \in \calK_\infty$) if it is strictly increasing, $\alpha(0) = 0$, and $\lim_{r \to \infty} \alpha(r) = \infty$.
\end{definition}
\begin{definition}[\cite{ames2019CBF}]
    Consider a function $h : \bbR^n \to \bbR$ which describes a set $S$ as in \eqref{eqn:S}-\eqref{eqn:S_interior}. Then, $h$ is a \textit{control barrier function (CBF)} if there exists an $\alpha \in \calK_\infty$ such that, for every $\vec{x} \in S$, there exists a $\vec{u} \in \bbR^m$ satisfying
    \begin{equation} \label{eqn:cbf_def}
        \frac{\partial h}{\partial\vec{x}}\Big|_{\vec{x}}(\vec{f}(\vec{x}) + G(\vec{x})\vec{u}) \geq -\alpha(h(\vec{x})).
    \end{equation}
\end{definition}
\begin{definition}[\cite{agrawal2021ICCBFs}] \label{def:ICCBF}
    Consider a function $h : \bbR^n \to \bbR$ which describes a set $S$ as in \eqref{eqn:S}-\eqref{eqn:S_interior}. Then, $h$ is a \textit{input-constrained control barrier function (ICCBF)} if there exists an $\alpha \in \calK_\infty$ such that, for every $\vec{x} \in S$, there exists a $\vec{u} \in \calU$ satisfying \eqref{eqn:cbf_def}.
\end{definition}
\begin{definition}[\cite{ames2019CBF}]
    A control policy $\vec{u} = \vec{k}(\vec{x})$ renders the set $S$ \textit{forward-invariant} if, for any $\vec{x}(0) \in S$, the policy results in $\vec{x}(t) \in S\ \forall t \geq 0$.
\end{definition}
\begin{definition}[\cite{ames2019CBF}]
    A control policy $\vec{u} = \vec{k}(\vec{x})$ renders the closed-loop dynamical system \textit{safe with respect to $S$} if it renders $S$ forward-invariant.
\end{definition}

In addition to the common definitions above, we introduce the following definitions. Consider selection of the control input $\vec{u}(t)$ by the following pointwise-in-time optimization problem:
\begin{subequations} \label{eqn:safety_filter}
    \begin{gather}
        \vec{u}(t) = \argminbelow_{\vec{u} \in \bbR^m} \|\vec{u} - \vec{u}_{nom}(t)\|^2 \label{eqn:safety_filter_cost} \\
        \mathrm{s.t.} \nonumber \\
        \vec{b}_1(\vec{x}(t))^\top\vec{u} + c_1(\vec{x}(t)) \geq 0, \label{eqn:safety_filter_constraint_1} \\
        \vdots \nonumber \\
        \vec{b}_N(\vec{x}(t))^\top\vec{u} + c_N(\vec{x}(t)) \geq 0, \label{eqn:safety_filter_constraint_N} \\
        \vec{u} \in \calU \label{eqn:safety_filter_input_constraint}
    \end{gather}
\end{subequations}
where $\vec{u}_{nom}(t)$ is any nominal control input, and $\vec{b}_k : \bbR^n \to \bbR^m$ and $c_k : \bbR^n \to \bbR$, $k \in \{1, \dots, N\}$ are continuous functions. Note that if $\calU$ is the intersection of half-spaces, \eqref{eqn:safety_filter_cost}-\eqref{eqn:safety_filter_input_constraint} is a quadratic program.

For the following definitions, consider a composite set $\bar{S} = \bigcap_{\ell=1}^p S_\ell \subset \bbR^n$ where each $S_\ell$ is defined by a smooth function $h_\ell(\vec{x})$ as in \eqref{eqn:S}-\eqref{eqn:S_interior}.
\begin{definition} \label{def:safety_filter}
    The filter in \eqref{eqn:safety_filter_cost}-\eqref{eqn:safety_filter_input_constraint} is a \textit{safety filter} for the plant in \eqref{eqn:prelim_plant} with respect to $\bar{S}$ if there exist $\alpha_\ell \in \calK_\infty$ such that, for any $\vec{x} \in \bar{S}$ such that the optimization problem is feasible, choosing $\vec{u}$ according to \eqref{eqn:safety_filter} ensures that
    \begin{equation}
        \frac{d}{dt}h_\ell(\vec{x}) \geq -\alpha_\ell(h(\vec{x}))
    \end{equation}
    for every $\ell \in \{1, \dots, p\}$.
\end{definition}
\begin{definition} \label{def:implementability}
    A safety filter for the plant in \eqref{eqn:prelim_plant} with respect to $\bar{S}$ is \textit{implementable} if the optimization problem is feasible for every $\vec{x} \in \bar{S}$.
\end{definition}
% \begin{definition}
%     For an initial condition $\vec{x}(0)$, multiple safety filters $\vec{b}_k(\vec{x})^\top\vec{u} + c_k(\vec{x}) \geq 0$, $1 \leq k \leq p$ are \textit{compatible} if, applying all filters simultaneously starting at $t = 0$, the resulting trajectory $\vec{x}(t)$ is such that there always exists a $\vec{u}(t) \in \calU$ satisfying all filters simultaneously at every $t \geq 0$.
% \end{definition}

% Intuitively (and with a slight abuse of terminology), implementability of a safety filter can be thought of as forward-invariance of the filter's feasibility. As long as an implementable safety filter is applied continuously starting at $t = 0$, it will remain feasible for all time.
% Compatibility of multiple filters implies implementability of each filter individually, but is a stronger condition, requiring that all filters be implementable simultaneously.

\begin{remark}
    For $S$ and $h(\vec{x})$ in \eqref{eqn:S}-\eqref{eqn:S_interior}, suppose that $h$ a CBF with relative degree 1 (see e.g. \cite{xiao2019HOCBF}). Then, with $N = 1$, \eqref{eqn:safety_filter_cost}-\eqref{eqn:safety_filter_input_constraint} is a safety filter for the plant in \eqref{eqn:prelim_plant} with respect to $S$ if $\vec{b}_1(\vec{x})^\top = \frac{\partial h}{\partial\vec{x}}|_{\vec{x}}G(\vec{x})$ and $c_1(\vec{x}) = \alpha(h(\vec{x})) + \frac{\partial h}{\partial\vec{x}}|_{\vec{x}}\vec{f}(\vec{x})$. Furthermore, implementability of the safety filter is equivalent to $h$ being an ICCBF. The definitions of safety filters and implementability here are useful in generalizing to settings with high relative degree and/or multiple simultaneous safe sets.
\end{remark}

\subsection{A Useful Technical Lemma}

We now present an important technical lemma using the preceeding definitions.
\begin{lemma} \label{lem:key_technical_lemma}
    Consider a collection of sets $S_1, \dots, S_p \subset \bbR^n$ defined by smooth functions $h_1, \dots, h_p : \bbR^n \to \bbR$ as in \eqref{eqn:S}-\eqref{eqn:S_interior}, and define the safe set $\bar{S} = \bigcap_{\ell=1}^p S_\ell$. Suppose that, for some choice of $\vec{b}_k$ and $c_k$, $k \in \{1, \dots, N\}$, \eqref{eqn:safety_filter_cost}-\eqref{eqn:safety_filter_input_constraint} is an implementable safety filter for the input-constrained plant in \eqref{eqn:prelim_plant}-\eqref{eqn:input_constraint} with respect to $\bar{S}$. Then, choosing $\vec{u}(t)$ as the solution to \eqref{eqn:safety_filter_cost}-\eqref{eqn:safety_filter_input_constraint} for any $\vec{u}_{nom}(t)$ renders the closed-loop dynamical system safe with respect to $\bar{S}$.
\end{lemma}
\begin{proof}
    See Appendix \ref{app:technical_lemma_proof}.
\end{proof}
% \begin{remark}
%     Condition 3) in Lemma \ref{lem:key_tech_lemma} is carefully worded to apply to the high-relative-degree setting where $\frac{\partial h}{\partial\vec{x}}G(\vec{x}) = 0$. In such cases, a backstepping approach is typically used to create a safety filter: see e.g. \cite{}.
% \end{remark}
% \begin{remark}
%     Lemma \ref{lem:key_tech_lemma} is related to the approach taken in \cite{agrawal2021ICCBFs} in that for a single safe set, the approach in \cite{agrawal2021ICCBFs} will create a safety filter satisfying conditions 1) and 3). The main difference is that we don't explicitly ask for a safe set (or subset) in which the filter is satisfiable everywhere, and we instead ask for a filter which implicitly restricts $\vec{x}(t)$ to the parts of the safe set where it stays satisfiable. This change of mindset allows us to consider analytic approaches not based on explicit construction of sets and subsets.
% \end{remark}

With Lemma \ref{lem:key_technical_lemma} established, the remainder of this paper will focus on analytic methods for obtaining implementable safety filters for certain classes of problems.
\section{Problem Statements} \label{sec:problem_statements}

% We revisit the problem setting considered in \cite{doeser2020invariant}, which is control of the $n$th-order integrator
We consider a dynamical system consisting of $n$ chained integrators:
\begin{equation} \label{eqn:plant_scalar_integrator}
    \dot{x}_1 = x_2,\; \dot{x}_2 = x_3,\; \cdots,\; \dot{x}_{n-1} = x_n,\; \dot{x}_n = u
\end{equation}
subject to the input constraint
\begin{equation} \label{eqn:input_constraints}
    u \in \calU \subset \bbR,\; \calU = \{u \in \bbR : \ubar{u} \leq u \leq \bar{u}\}
\end{equation}
where $\ubar{u} < 0 < \bar{u}$. Let $\vec{x} = [x_1, x_2, \cdots, x_n]^\top$.

We consider the following two problem statements:
\begin{problem}[Simplified] \label{plm:simplified}
    % Choose $u(t)$ to ensure that the input-constrained plant in \eqref{eqn:plant_scalar_integrator}-\eqref{eqn:input_constraints} satisfies the state constraint
    Determine a $u(t)$ for the plant in \eqref{eqn:plant_scalar_integrator}-\eqref{eqn:input_constraints} such that the state constraint
    \begin{equation} \label{eqn:state_constraints_simple}
        \vec{x}(t) \in S \subset \bbR^n,\; S = \{\vec{x} \in \bbR^n : x_1 \geq \ubar{x}_1\}
    \end{equation}
    is forward-invariant.
\end{problem}
\begin{problem}[Full] \label{plm:full}
    % Choose $u(t)$ to ensure that the input-constrained plant in \eqref{eqn:plant_scalar_integrator}-\eqref{eqn:input_constraints} satisfies the state constraints
    Determine a $u(t)$ for the plant in \eqref{eqn:plant_scalar_integrator}-\eqref{eqn:input_constraints} such that the state constraints
    \begin{equation} \label{eqn:state_constraints_full}
        \begin{gathered}
        \vec{x}(t) \in S_1 \cap S_2 \cap \cdots \cap S_n \subset \bbR^n, \\
        S_j = \{\vec{x} \in \bbR^n : \ubar{x}_j \leq x_j \leq \bar{x}_j\}
    \end{gathered}
    \end{equation}
    are forward-invariant, where $\ubar{x}_j < 0 < \bar{x}_j\ \forall j \geq 2$.
\end{problem}

We will first consider the easier problem in Problem \ref{plm:simplified} in order to lay out a principled method of constraining $u(t)$ to satisfy simultaneous state and input constraints. We will then show how to extend the approach derived for the simplified problem to handle the full problem in Problem \ref{plm:full}.
\section{Solution to Problem \ref{plm:simplified}} \label{sec:simplified_solution}

We first address the simplifed problem statement in Problem \ref{plm:simplified}. A full derivation of the algorithm, complete with motivations for the design choices, can be found in Appendix \ref{app:ICCBF_derivation}.

Consider the following series of recursively-defined functions:
\begin{subequations} \label{eqn:simplified_function_recursion}
    \begin{align}
        h_1(\vec{x}) &= x_1 - \ubar{x}_1, \; \Delta_1 = 0, \; \gamma_0 = 0, \label{eqn:simplified_h_1} \\
        h_i(\vec{x}) &= x_i + \gamma_{i-1}\sqrt{h_{i-1}} - \frac{\gamma_{i-2}^2}{2} - \epsilon_{i-1}\ \forall i \in \{2, \dots, n\}, \label{eqn:simplified_h_i} \\
        \Delta_i(\vec{x}) &= \frac{\gamma_{i-1}}{2\sqrt{h_{i-1}}}(h_i + \Delta_{i-1} + \epsilon_{i-1})\ \forall i \in \{2, \dots, n-1\} \label{eqn:simplified_Delta_i}
    \end{align}
\end{subequations}
with some constant parameters $\gamma_i, \epsilon_i > 0$.
% Then, suppose that $u(t)$ is chosen to satisfy the constraint
% \small
% \begin{equation} \label{eqn:ICCBF_filter}
%     u \geq -\gamma_n\sqrt{h_n} - \frac{\gamma_{n-1}}{2\sqrt{h_{n-1}}}(h_n + \Delta_{n-1} + \epsilon_{n-1}) + \frac{\gamma_{n-1}^2}{2} + \epsilon_n,
% \end{equation}
% \normalsize
% which can be done by choosing $u(t)$ according to \eqref{eqn:safety_filter} for any $u_{nom}(t)$ with $N = 1$, $b_1 = 1$, and
% \small
% \begin{equation} \label{eqn:simplified_c}
%     c_1 = \gamma_n\sqrt{h_n} + \frac{\gamma_{n-1}}{2\sqrt{h_{n-1}}}(h_n + \Delta_{n-1} + \epsilon_{n-1}) - \frac{\gamma_{n-1}^2}{2} - \epsilon_n.
% \end{equation}
% \normalsize
Then, consider the relative-degree-$i$ safe sets
\begin{equation} \label{eqn:simplified_safe_subsets}
    S_i = \{\vec{x} \in \bbR^n : h_i(\vec{x}) \geq \epsilon_i^2/\gamma_i^2\}
\end{equation}
and the composite safe set
\begin{equation} \label{eqn:simplified_composite_safe_set}
    \bar{S} = \bigcap_{i=1}^n S_i,
\end{equation}
and note that $\bar{S} \subset S$ for $S$ defined in \eqref{eqn:state_constraints_simple}. The following is our main result for Problem \ref{plm:simplified}.
\begin{theorem} \label{thm:simplified_solution}
    Consider the problem statement in Problem \ref{plm:simplified}. Define $h_1, \dots, h_n$ and $\Delta_2, \dots, \Delta_{n-1}$ as in \eqref{eqn:simplified_h_1}-\eqref{eqn:simplified_Delta_i} with $\gamma_i, \epsilon_i > 0$ and $\gamma_{n-1} \leq \sqrt{2\bar{u}}$. Then, \eqref{eqn:safety_filter_cost}-\eqref{eqn:safety_filter_input_constraint} with $N = 1$, $b_1 = 1$,
    % \small
    \begin{equation} \label{eqn:simplified_c}
        \begin{aligned}
            c_1 =&\ \gamma_n\sqrt{h_n} - \frac{\gamma_{n-1}^2}{2} - \epsilon_n + \frac{\gamma_{n-1}}{2\sqrt{h_{n-1}}}(h_n + \Delta_{n-1} + \epsilon_{n-1}),
        \end{aligned}
    \end{equation}
    % \normalsize
    is an implementable safety filter for the plant in \eqref{eqn:plant_scalar_integrator} with respect to $\bar{S}$ in \eqref{eqn:simplified_composite_safe_set}.
\end{theorem}
\begin{proof}
    See Appendix \ref{app:simplified_proof}.
\end{proof}
\begin{corollary} \label{cor:simplified_solution}
    If $\vec{x}(0) \in \bar{S}$ and $u(t)$ is chosen as the solution to the safety filter described in Theorem \ref{thm:simplified_solution} for any $u_{nom}(t)$, then $\vec{x}(t) \in S\ \forall t \geq 0$ for $S$ in \eqref{eqn:state_constraints_simple}.
\end{corollary}
\begin{proof}
    From Theorem \ref{thm:simplified_solution} and Lemma \ref{lem:key_technical_lemma}, it immediately follows that $\bar{S}$ is forward-invariant. The result follows from the fact that $\bar{S} \subset S$.
\end{proof}

\subsection{Intuition: Why is Square Root the Magic Function?} \label{subsec:intuition_double_integrator}

For intuition regarding the choice of $h_i$, particularly for the use of the square root function, consider the special case of Problem \ref{plm:simplified} with $n = 2$. We can think of this case as applying a force to a particle of unit mass to prevent it from crossing a barrier at $x_1 = \ubar{x}_1$. If the particle is not moving toward the barrier at time $t$, i.e. $x_2(t) \geq 0$, then $|x_2(t)|$ can be arbitrarily large without endangering the particle. However, suppose the particle is moving toward the barrier, i.e. $x_2(t) < 0$. Then, for any value of $x_1(t)$, it is possible for $|x_2(t)|$ to be large enough that the particle cannot avoid colliding even under maximum deceleration, i.e. $u(t) = \bar{u}$.

Quantitatively, we can define the particle's minimum stopping distance at time $t$, $x_{stop}(t)$, as the distance from the barrier such that, if $x_2(t) < 0$ and $u(\tau) = \bar{u} \ \forall \tau \geq t$, the particle will eventually arrive at the barrier with zero speed. It can be shown that the minimum stopping distance at time $t$ is given by
\begin{equation} \label{eqn:min_stopping_distance}
    x_{stop}(t) = \frac{x_2(t)^2}{2\bar{u}}.
\end{equation}
Then, we can prevent the particle from crossing the barrier at $x_1 = \ubar{x}_1$ by enforcing the state constraint
\begin{equation} \label{eqn:min_stopping_dist_criterion_1}
    \begin{cases} x_1(t) - \ubar{x}_1 \geq 0, & x_2(t) \geq 0 \\ x_1(t) - \ubar{x}_1 \geq x_{stop}(t), & x_2(t) < 0 \end{cases}\ \forall t \geq 0
\end{equation}
which, assuming $x_1(0) \geq \ubar{x}_1$, is equivalent to
\begin{equation} \label{eqn:min_stopping_dist_criterion_2}
    x_2(t) \geq -\sqrt{2\bar{u}(x_1(t) - \ubar{x}_1)}\ \forall t \geq 0.
\end{equation}

As $x_{stop}$ is a nonnegative quantity, it is clear that $x_1 - \ubar{x}_1 \geq x_{stop}$ implies $x_1 \geq \ubar{x}_1$. Furthermore, by the definition of $x_{stop}$, at every time $t$ we know that applying $u(\tau) = \bar{u}\ \forall \tau \geq t$ will ensure that the particle never crosses the barrier. Rearranging \eqref{eqn:min_stopping_dist_criterion_2}, it is thus clear that
\begin{equation} \label{eqn:double_integrator_ICCBF}
    h_2(\vec{x}) = x_2 + \sqrt{2\bar{u}(x_1 - \ubar{x}_1)}
\end{equation}
is an ICCBF for the state constraint in \eqref{eqn:state_constraints_simple}.

Note that \eqref{eqn:double_integrator_ICCBF} is mostly the same $h_2$ as one would obtain from \eqref{eqn:simplified_h_1}-\eqref{eqn:simplified_Delta_i} with $\gamma_{n-1} = \sqrt{2\bar{u}}$. In \eqref{eqn:simplified_h_1}-\eqref{eqn:simplified_Delta_i}, we add an additional parameter $\epsilon_1 > 0$ which prevents $x_1 - \ubar{x}_1$ from going all the way to zero for numerical stability. Thus, one can think of \eqref{eqn:simplified_h_1}-\eqref{eqn:simplified_Delta_i} as an iterated minimum stopping distance constraint for an $n$th-order integrator, where the $\Delta$s are merely a modification to prevent the number of terms in each $h_i$ from quadratically growing for $i \geq 3$.
\section{Solution to Problem \ref{plm:full}} \label{sec:full_solution}

We now turn our attention to the full problem statement in Problem \ref{plm:full}. As the problem statement contains $2n$ simultaneous state constraints, with $n$ lower bounds and $n$ upper bounds, we must consider $2n$ series of recursively-defined functions. For each $j \in \{1, \dots, n\}$, define:
\begin{subequations} \label{eqn:full_lower_function_recursion}
    \begin{align}
        \ubar{h}_{j1}(\vec{x}) &= x_j - \ubar{x}_j, \; \ubar{\Delta}_{j1} = 0, \; \ubar{\gamma}_{j0} = 0, \label{eqn:lower_h_1} \\
        \ubar{h}_{ji}(\vec{x}) &= x_{i+j-1} + \ubar{\gamma}_{j(i-1)}\sqrt{\ubar{h}_{j(i-1)}} - \frac{\ubar{\gamma}_{j(i-2)}^2}{2} - \ubar{\epsilon}_{j(i-1)}\ \forall i \in \{2, \dots, n-j+1\}, \label{eqn:lower_h_i} \\
        \ubar{\Delta}_{ji}(\vec{x}) &= \frac{\ubar{\gamma}_{j(i-1)}}{2\sqrt{\ubar{h}_{j(i-1)}}}(\ubar{h}_{ji} + \ubar{\Delta}_{j(i-1)} + \ubar{\epsilon}_{j(i-1)})\ \forall i \in \{2, \dots, n-j\} \label{eqn:lower_Delta_i}
    \end{align}
\end{subequations}
and
\begin{subequations} \label{eqn:full_upper_function_recursion}
    \begin{align}
        \bar{h}_{j1}(\vec{x}) &= \bar{x}_j - x_j, \; \bar{\Delta}_{j1} = 0, \; \bar{\gamma}_{j0} = 0, \label{eqn:upper_h_1} \\
        \bar{h}_{ji}(\vec{x}) &= -x_{i+j-1} + \bar{\gamma}_{j(i-1)}\sqrt{\bar{h}_{j(i-1)}} - \frac{\bar{\gamma}_{j(i-2)}^2}{2} - \bar{\epsilon}_{j(i-1)}\ \forall i \in \{2, \dots, n-j+1\}, \label{eqn:upper_h_i} \\
        \bar{\Delta}_{ji}(\vec{x}) &= \frac{\bar{\gamma}_{j(i-1)}}{2\sqrt{\bar{h}_{j(i-1)}}}(\bar{h}_{ji} + \bar{\Delta}_{j(i-1)} + \bar{\epsilon}_{j(i-1)})\ \forall i \in \{2, \dots, n-j\} \label{eqn:upper_Delta_i}
    \end{align}
\end{subequations}
with some constant parameters $\ubar{\gamma}_{ji}, \ubar{\epsilon}_{ji}, \bar{\gamma}_{ji}, \bar{\epsilon}_{ji} > 0$.
% We then produce the following $2n$ simultaneous safety filters from $\ubar{h}_{j(n-j+1)}$ and $\bar{h}_{j(n-j+1)}$:
% \small
% \begin{equation} \label{eqn:lower_ICCBF_filter}
% \begin{aligned}
%     u \geq &-\ubar{\gamma}_{j(n-j+1)}\sqrt{\ubar{h}_{j(n-j+1)}} \\
%     &- \frac{\ubar{\gamma}_{j(n-j)}}{2\sqrt{\ubar{h}_{j(n-j)}}}(\ubar{h}_{j(n-j+1)} + \ubar{\Delta}_{j(n-j)} + \ubar{\epsilon}_{j(n-j)}) \\
%     &+ \frac{\ubar{\gamma}_{j(n-j)}^2}{2} + \ubar{\epsilon}_{j(n-j+1)},
% \end{aligned}
% \end{equation}
% \begin{equation} \label{eqn:upper_ICCBF_filter}
% \begin{aligned}
%     u \leq\ &\bar{\gamma}_{j(n-j+1)}\sqrt{\bar{h}_{j(n-j+1)}} \\
%     &+ \frac{\bar{\gamma}_{j(n-j)}}{2\sqrt{\bar{h}_{j(n-j)}}}(\bar{h}_{j(n-j+1)} + \bar{\Delta}_{j(n-j)} + \bar{\epsilon}_{j(n-j)}) \\
%     &- \frac{\bar{\gamma}_{j(n-j)}^2}{2} - \bar{\epsilon}_{j(n-j+1)}.
% \end{aligned}
% \end{equation}
% \normalsize
Then, consider the relative-degree-$i$ safe sets
\begin{subequations} \label{eqn:full_safe_subsets}
    \begin{align}
        \ubar{S}_{ji} &= \{\vec{x} \in \bbR^n : \ubar{h}_{ji}(\vec{x}) \geq \ubar{\epsilon}_{ji}^2/\ubar{\gamma}_{ji}^2\}, \label{eqn:full_safe_subsets_lower} \\
        \bar{S}_{ji} &= \{\vec{x} \in \bbR^n : \bar{h}_{ji}(\vec{x}) \geq \bar{\epsilon}_{ji}^2/\bar{\gamma}_{ji}^2\} \label{eqn:full_safe_subsets_upper}
    \end{align}
\end{subequations}
and the composite safe set
\begin{equation} \label{eqn:full_composite_safe_set}
    \bar{S} = \bigcap_{j=1}^n\bigcap_{i=1}^{n-j+1} (\ubar{S}_{ji} \cap \bar{S}_{ji}),
\end{equation}
and note that $\bar{S} \subset \bigcap_{j=1}^n S_j$ for $S_j$ defined in \eqref{eqn:state_constraints_full}.

\subsection{Parameter Selection} \label{subsec:full_parameter_selection}

In order to realize an implementable safety filter with $N = 2n$ simultaneous constraints, we need to choose the $2n(n+1)$ parameters $\ubar{\gamma}_{ji}, \ubar{\epsilon}_{ji}, \bar{\gamma}_{ji}, \bar{\epsilon}_{ji}$ appropriately. This can be most easily accomplished by reparametrization. First, reduce the number of parameters from $2n(n+1)$ to $2n$ by choosing, for each $j \in \{1, \dots, n\}$,
\begin{equation} \label{eqn:gamma_epsilon_matching}
\begin{gathered}
    \ubar{\gamma}_{ji} = \bar{\gamma}_{ji} = \gamma_{i+j-1},\quad \ubar{\epsilon}_{ji} = \bar{\epsilon}_{ji} = \epsilon_{i+j-1} \quad \forall i \in \{1, \dots, n-j+1\}
\end{gathered}
\end{equation}
for some $\gamma_1, \dots, \gamma_n, \epsilon_1, \dots, \epsilon_n > 0$ to be calculated. Then, choose any $\delta \in (0, \bar{x}_1 - \ubar{x}_1)$, any $\alpha_2, \dots, \alpha_n > 0$, and any $\beta_1, \dots, \beta_n, \eta_1, \dots, \eta_{n-1} \in (0, 1)$. Finally, we will seek $\gamma_1, \dots, \gamma_n, \epsilon_1, \dots, \epsilon_n$ satisfying the following set of inequalities:
\begin{subequations}
    \begin{align}
        &\gamma_1 > 0 \label{eqn:gamma_1_lower_limit} \\
        &\gamma_2 \geq \textstyle \max\{\frac{(1/2 + \alpha_2)\gamma_1^{3/2}}{\sqrt{\beta_1\sqrt{\delta}}}, \frac{(1 + \alpha_2)\gamma_1^{3/2}}{\sqrt{\beta_1\sqrt{\bar{x}_1 - \ubar{x}_1}}}, \frac{(1/2 + \alpha_2)\gamma_1^2}{\sqrt{\min\{-\ubar{x}_2, \bar{x}_2\}}}\}, \label{eqn:gamma_2_lower_limit} \\
        &\gamma_i \geq \textstyle \max\{\frac{(1 + \alpha_i)\gamma_{i-1}^2}{\gamma_{i-2}\sqrt{\beta_{i-1}\alpha_{i-1}}}, \frac{(1/2 + \alpha_i)\gamma_{i-1}^{3/2}}{\sqrt{\beta_{i-1}\sqrt{\min\{-\ubar{x}_{i-1}, \bar{x}_{i-1}\}}}}, \frac{(1 + \alpha_i)\gamma_{i-1}^{3/2}}{\sqrt{\beta_{i-1}\sqrt{\bar{x}_{i-1} - \ubar{x}_{i-1}}}}, \frac{(1/2 + \alpha_i)\gamma_{i-1}^2}{\sqrt{\eta_{i-2}\min\{-\ubar{x}_i, \bar{x}_i\}}}\}\ \forall i \in \{3, \dots, n\}, \label{eqn:gamma_i_lower_limit} \\
        &\epsilon_1 \leq \textstyle (1 - \beta_1)\min\{\gamma_1\sqrt{\delta}, \frac{\gamma_1\sqrt{\bar{x}_1 - \ubar{x}_1}}{2}\}, \label{eqn:epsilon_1_upper_limit} \\
        &\epsilon_i \leq \textstyle (1 - \beta_i)\min\{\frac{\alpha_i\gamma_{i-1}^2}{2}, \gamma_i\sqrt{\min\{-\ubar{x}_i, \bar{x}_i\}}, \frac{\gamma_i\sqrt{\bar{x}_i - \ubar{x}_i}}{2}\}\ \forall i \in \{2, \dots, n\}, \label{eqn:epsilon_i_upper_limit} \\
    % \end{align}
    % \begin{align}
        &\gamma_i \leq \sqrt{2(1 - \eta_i)\min\{-\ubar{x}_{i+2}, \bar{x}_{i+2}\}}\ \forall i \in \{1, \dots, n-2\}, \label{eqn:gamma_i_upper_limit} \\
        &\gamma_{n-1} \leq \sqrt{2(1 - \eta_{n-1})\min\{-\ubar{u}, \bar{u}\}} \label{eqn:gamma_n-1_upper_limit}
    \end{align}
\end{subequations}

Appendix \ref{app:parameter_tuning} presents an algorithm which, for given values of $\delta$, $\alpha_i$, $\beta_i$, and $\eta_i$, seeks the set of parameters $\gamma_1, \dots, \gamma_n, \epsilon_1, \dots, \epsilon_n$ satisfying \eqref{eqn:gamma_1_lower_limit}-\eqref{eqn:gamma_i_upper_limit} with the largest possible value of$\gamma_1$. Our tuning algorithm has the following property:
\begin{lemma} \label{lem:tuning}
    Algorithm \ref{alg:parameter_tuning} returns a set of parameters $\gamma_1, \dots, \gamma_n, \epsilon_1, \dots, \epsilon_n$ which satisfy \eqref{eqn:gamma_1_lower_limit}-\eqref{eqn:gamma_n-1_upper_limit} in finitely many steps.
\end{lemma}
\begin{proof}
    See Appendix \ref{app:tuning_proof}.
\end{proof}
\begin{remark}
    Algorithm \ref{alg:parameter_tuning} will tune the CBF parameters such that the safety filter considers any location $x_1 \in [\ubar{x}_1 + \delta, \bar{x}_1 - \delta]$ to be a permissible ``fixed point." In other words, whenever $x_1 \in [\ubar{x}_1 + \delta, \bar{x}_1 - \delta]$ and $x_2 = \cdots = x_n = 0$, the safety filter will treat $u = 0$ as a feasible control input and allow the closed-loop system to hold its current position. Thus, $\delta$ should be chosen such that $[\ubar{x}_1 + \delta, \bar{x}_1 - \delta]$ is the normal operating range for a given application. Note, however, that $[\ubar{x}_1 + \delta, \bar{x}_1 - \delta]$ is not guaranteed to be forward-invariant; only $[\ubar{x}_1, \bar{x}_1]$ is guaranteed to be forward-invariant.
\end{remark}

\subsection{Main Result}

The following is our main result for Problem \ref{plm:full}.
\begin{theorem} \label{thm:full_solution}
    Consider the problem statement in Problem \ref{plm:full}. Choose any $\delta \in (0, \bar{x}_1 - \ubar{x}_1)$, $\alpha_2, \dots, \alpha_n > 0$, and $\beta_1, \dots, \beta_n, \eta_1, \dots, \eta_{n-1} \in (0, 1)$ and obtain parameters $\gamma_1, \dots, \gamma_n$ and $\epsilon_1, \dots, \epsilon_n$ from Algorithm \ref{alg:parameter_tuning}. Then, for each $j \in \{1, \dots, n\}$, define $\ubar{h}_{j1}, \dots, \ubar{h}_{j(n-j+1)}$, $\ubar{\Delta}_{j2}, \dots, \ubar{\Delta}_{j(n-j)}$, $\bar{h}_{j1}, \dots, \bar{h}_{j(n-j+1)}$, and $\bar{\Delta}_{j2}, \dots, \bar{\Delta}_{j(n-j)}$ as in \eqref{eqn:lower_h_1}-\eqref{eqn:upper_Delta_i} with $\ubar{\gamma}_{ji}, \ubar{\epsilon}_{ji}, \bar{\gamma}_{ji}, \bar{\epsilon}_{ji}$ chosen according to \eqref{eqn:gamma_epsilon_matching}. For each $j \in \{1, \dots, n\}$, set $b_j$, $b_{n+j}$, $c_j$, and $c_{n+j}$ in \eqref{eqn:safety_filter_cost}-\eqref{eqn:safety_filter_input_constraint} as
    % {\small
    \begin{subequations}
        \begin{align}
            &b_j =\ 1, \label{eqn:full_b_lower} \\
            &c_j =\ \ubar{\gamma}_{j(n-j+1)}\sqrt{\ubar{h}_{j(n-j+1)}} - \frac{\ubar{\gamma}_{j(n-j)}^2}{2} - \ubar{\epsilon}_{j(n-j+1)} + \frac{\ubar{\gamma}_{j(n-j)}}{2\sqrt{\ubar{h}_{j(n-j)}}}(\ubar{h}_{j(n-j+1)} + \ubar{\Delta}_{j(n-j)} + \ubar{\epsilon}_{j(n-j)}), \label{eqn:full_c_lower} \\
            &b_{n+j} =\ -1, \label{eqn:full_b_upper} \\
            &c_{n+j} =\ \bar{\gamma}_{j(n-j+1)}\sqrt{\bar{h}_{j(n-j+1)}} - \frac{\bar{\gamma}_{j(n-j)}^2}{2} - \bar{\epsilon}_{j(n-j+1)} + \frac{\bar{\gamma}_{j(n-j)}}{2\sqrt{\bar{h}_{j(n-j)}}}(\bar{h}_{j(n-j+1)} + \bar{\Delta}_{j(n-j)} + \bar{\epsilon}_{j(n-j)}). \label{eqn:full_c_upper}
        \end{align}
    \end{subequations}
    % }
    Then, \eqref{eqn:safety_filter_cost}-\eqref{eqn:safety_filter_input_constraint} with $N = 2n$ is an implementable safety filter for the plant in \eqref{eqn:plant_scalar_integrator} with respect to $\bar{S}$ in \eqref{eqn:full_composite_safe_set}.
\end{theorem}
\begin{proof}
    See Appendix \ref{app:full_proof}.
\end{proof}
\begin{corollary} \label{cor:full_solution}
    If $\vec{x}(0) \in \bar{S}$ and $u(t)$ is chosen as the solution to the safety filter described in Theorem \ref{thm:full_solution} for any $u_{nom}(t)$, then $\vec{x}(t) \in \bigcap_{j=1}^n S_j\ \forall t \geq 0$ for $S_j$ in \eqref{eqn:state_constraints_full}.
\end{corollary}
\begin{proof}
    Analogous to the proof of Corollary \ref{cor:simplified_solution}.
\end{proof}
\section{Extension to the Multi-Input Setting} \label{sec:extension}

In this section, we discuss an extension of our main problem statement to the multi-input setting. For simplicity, we only consider simultaneous state constraints with the same relative degree, not with varying relative degrees as in Problem \ref{plm:full}.

\subsection{Extended Problem Statement} \label{subsec:extended_problem_statement}

In this section, we consider control of the $m$-dimensional $n$th-order integrator
\begin{equation} \label{eqn:plant_md_integrator}
    \dot{\vec{x}}_1 = \vec{x}_2, \dot{\vec{x}}_2 = \vec{x}_3, \cdots, \dot{\vec{x}}_{n-1} = \vec{x}_n, \dot{\vec{x}}_n = \vec{u}
\end{equation}
where $\vec{x}_j, \vec{u} \in \bbR^m$, subject to the input constraint
\begin{equation} \label{eqn:md_input_constraint}
    \vec{u} \in \calU \subset \bbR^m,\; \calU = \{\vec{u} \in \bbR^m : \|\vec{u}\| \leq \bar{u}\}
\end{equation}
where $\bar{u} > 0$. Let $\vec{x} = [\vec{x}_1^\top, \cdots, \vec{x}_n^\top]^\top$. The extended problem statement is as follows:
\begin{problem}[Extended] \label{plm:extended}
    % Choose $\vec{u}(t)$ to ensure that the input-constrained plant in \eqref{eqn:plant_md_integrator}-\eqref{eqn:md_input_constraint} satisfies the state constraints
    Determine a $\vec{u}(t)$ for the plant in \eqref{eqn:plant_md_integrator}-\eqref{eqn:md_input_constraint} such that the state constraints
    \begin{equation} \label{eqn:md_state_constraints}
        \begin{gathered}
            \vec{x} \in S_1 \cap S_2 \cap \cdots \cap S_p \subset \bbR^{nm}, \\
            S_k = \{\vec{x} \in \bbR^{nm} : \vec{r}_k^\top\vec{x}_1 + s_k \geq 0\}
        \end{gathered}
    \end{equation}
    are forward-invariant, where $\vec{r}_k$ and $s_k$ are such that $S_1 \cap S_2 \cap \cdots \cap S_p$ is nonempty and has positive $nm$-dimensional volume.
\end{problem}
Note that $S_1 \cap S_2 \cap \cdots \cap S_p$ is closed by definition, but we do not require it to be bounded.

\subsection{Solution Setup for Problem \ref{plm:extended}}

We now present our approach to constructing a safety filter for Problem \ref{plm:extended}. The problem statement now contains $p$ simultaneous state constraints, so we now consider the following $p$ series of recursively-defined functions. For each $k \in \{1, \dots, p\}$, define:
\begin{subequations}
    \begin{align}
        h_{k1}(\vec{x}) &= \vec{r}_k^\top\vec{x}_1 + s_k, \; \Delta_{k1} = 0, \; \gamma_{k0} = 0, \label{eqn:extended_h_1} \\
        h_{ki}(\vec{x}) &= \vec{r}_k^\top\vec{x}_i + \gamma_{k(i-1)}\sqrt{h_{k(i-1)}} - \frac{\gamma_{k(i-2)}^2}{2} - \epsilon_{k(i-1)}\ \forall i \in \{2, \dots, n\}, \label{eqn:extended_h_i} \\
        \Delta_{ki}(\vec{x}) &= \frac{\gamma_{k(i-1)}}{2\sqrt{h_{k(i-1)}}}(h_{ki} + \Delta_{k(i-1)} + \epsilon_{k(i-1)})\ \forall i \in \{2, \dots, n-1\} \label{eqn:extended_Delta_i}
    \end{align}
\end{subequations}
with some constant parameters $\gamma_{ki}, \epsilon_{ki} > 0$. Then, consider the relative-degree-$i$ safe sets
\begin{equation} \label{eqn:extended_safe_subsets}
    S_{ki} = \{\vec{x} \in \bbR^n : h_{ki}(\vec{x}) \geq \epsilon_{ki}^2/\gamma_{ki}^2\}
\end{equation}
and the composite safe set
\begin{equation} \label{eqn:extended_composite_safe_set}
    \bar{S} = \bigcap_{k=1}^p\bigcap_{i=1}^{n} S_{ki},
\end{equation}
and note that $\bar{S} \subset \bigcap_{k=1}^p S_k$ for $S_k$ defined in \eqref{eqn:md_state_constraints}.

The corresponding safety filter will consist of \eqref{eqn:safety_filter_cost}-\eqref{eqn:safety_filter_input_constraint} with $N = p$, and for each $k \in \{1, \dots, p\}$,
\begin{subequations}
    \begin{align}
        \vec{b}_k =&\ \vec{r}_k, \label{eqn:extended_filter_b} \\
        c_k =&\ \gamma_{kn}\sqrt{h_{kn}} - \frac{\gamma_{k(n-1)}^2}{2} - \epsilon_{kn} + \frac{\gamma_{k(n-1)}}{2\sqrt{h_{k(n-1)}}}(h_{kn} + \Delta_{k(n-1)} + \epsilon_{k(n-1)}). \label{eqn:extended_filter_c}
    \end{align}
\end{subequations}

\subsection{Discussion of Implementability} \label{subsec:extended_discussion}

In order to realize an implementable safety filter with $N = p$ constraints, we need to choose the $2np$ parameters $\gamma_{ki}, \epsilon_{ki}$ appropriately. Unfortunately, this is less straightforward than parameter selection for Problem \ref{plm:full}. Since Problem \ref{plm:full} considered the case where the input was a scalar, in order to guarantee implementability of the safety filter, it was sufficient to merely consider pairwise feasibility of the constraints on $u$: every upper bound on $u$ had to be higher than every lower bound on $u$. Considering constraints pairwise led to a set of conditions on the parameters $\gamma_i$ and $\epsilon_i$ in \eqref{eqn:gamma_1_lower_limit}-\eqref{eqn:gamma_n-1_upper_limit}, and it was always possible to choose values for $\gamma_i$ and $\epsilon_i$ satisfying all conditions.

In the general multi-dimensional setting, merely considering constraints pairwise is no longer sufficient.
% The $m$-dimensional body described by the intersection of the constraints in \eqref{eqn:safety_filter_constraint_1}-\eqref{eqn:safety_filter_input_constraint} must have positive volume for all $\vec{x} \in \bar{S}$, and it is more difficult to obtain corresponding sufficient conditions on $\gamma_{ki}$ and $\epsilon_{ki}$.
There are two special cases where it is straightforward to obtain sufficient conditions for implementability:
\begin{enumerate}
    \item In the setting where there exists a unit vector $\hat{u} \in \bbR^m$ satisfying $\hat{u}^\top\vec{r}_k > 0\ \forall k \in \{1, \dots, p\}$, one can show that choosing $\gamma_{k(n-1)} \leq \sqrt{2\bar{u}\hat{u}^\top\vec{r}_k}\ \forall k$ results in an implementable safety filter, as $\vec{u} = \bar{u}\hat{u}$ becomes a feasible solution whenever $\vec{x} \in \bar{S}$. Note that this corresponds to the setting where $\bigcap_{k=1}^p S_k$ is not bounded, and is equivalent to the assumption of non-conflicting CBFs in \cite{breeden2023compositions}.
    \item In the setting where the safe set $\bigcap_{k=1}^p S_k$ is a rectangular prism in $m$ dimensions, the problem can be decomposed into $m$ simultaneous instances of Problem \ref{plm:full}.
\end{enumerate}

For the general problem statement in Problem \ref{plm:extended}, we present the following conjecture, whose proof is a matter for future work:
\begin{conjecture} \label{cnj:extended_solution}
    Consider the problem statement in Problem \ref{plm:extended}. There exist constant values $\gamma_{ki}, \epsilon_{ki} > 0\ \forall k, i$ such that \eqref{eqn:safety_filter_cost}-\eqref{eqn:safety_filter_input_constraint} with $N = p$ and $\vec{b}_k, c_k$ given by \eqref{eqn:extended_filter_b}-\eqref{eqn:extended_filter_c} is an implementable safety filter for the plant in \eqref{eqn:plant_md_integrator} with respect to $\bar{S}$ in \eqref{eqn:extended_composite_safe_set}.
\end{conjecture}

\section{Simulation: Quadrotor Safety} \label{sec:simulations}

% We now demonstrate our approach through two simulation scenarios: one which demonstrates our approach to Problem \ref{plm:full}, and another which demonstrates the ability of our approach to extend to Problem \ref{plm:extended}.
We now demonstrate the efficacy of our approach in extending to Problem \ref{plm:extended} through simulation of an input-constrained quadrotor drone navigating an hourglass-shaped room.

\subsection{Quadrotor Modeling} \label{subsec:quadrotor_modeling}

In the following simulations, we make use of the linearized 6-DOF quadrotor model in \cite{Annaswamy2023ACRL} given by
\begin{gather}
    \dot{x} = v_x,\; \dot{v}_x = g\theta,\; \dot{\theta} = q,\; \dot{q} = \frac{1}{I_y}\tau_y, \label{eqn:quadrotor_x_motion} \\
    \dot{y} = v_y,\; \dot{v}_y = -g\phi,\; \dot{\phi} = p,\; \dot{p} = \frac{1}{I_x}\tau_x, \label{eqn:quadrotor_y_motion} \\
    \dot{z} = v_z,\; \dot{v}_z = \frac{1}{m}F, \label{eqn:quadrotor_z_motion} \\
    \dot{\psi} = r,\; \dot{r} = \frac{1}{I_z}\tau_z \label{eqn:quadrotor_psi_motion}
\end{gather}
where $(x, y, z)$ is the position of the COM, $(v_x, v_y, v_z)$ is the velocity of the COM, $(\phi, \theta, \psi)$ are the roll, pitch, and yaw Euler angles, $(p, q, r)$ is the angular velocity, $F$ is the net vertical force on the drone, and $(\tau_x, \tau_y, \tau_z)$ are the net torques about each axis. $F$, $\tau_x$, $\tau_y$, and $\tau_z$ are related to the normalized thrusts from each rotor (thrust minus $\frac{mg}{4}$), $u_i$, as
\begin{equation} \label{eqn:quadrotor_forces_thrusts}
    \begin{bmatrix} F \\ \tau_y \\ \tau_x \\ \tau_z \end{bmatrix} = \begin{bmatrix} 1 & 1 & 1 & 1 \\ 0 & L & 0 & -L \\ L & 0 & -L & 0 \\ \nu & -\nu & \nu & -\nu \end{bmatrix}\begin{bmatrix} u_1 \\ u_2 \\ u_3 \\ u_4 \end{bmatrix} = B_2\vec{u}.
\end{equation}
Let $\vec{x} = [x, y, y, \theta, \phi, \psi, v_x, v_y, v_z, q, p, r]^\top$. Then, we can write
\begin{equation}
    \dot{\vec{x}} = A\vec{x} + B_1B_2\vec{u}
\end{equation}
where $A$ and $B_1$ are derived in a straightforward manner from \eqref{eqn:quadrotor_x_motion}-\eqref{eqn:quadrotor_psi_motion}. We further suppose that the quadrotor is subject to input saturation on each $u_i$ individually, i.e.
\begin{equation}
    |u_i| \leq \bar{u},\; i \in \{1, 2, 3, 4\}
\end{equation}
for some $\bar{u} > 0$.

\subsection{Simulation Scenario}

For this simulation, we consider an hourglass-shaped room (see Figure \ref{fig:traj_1_top_view}), and suppose that the drone must travel from its initial position to a target position while avoiding the walls. This safe set is non-convex and thus cannot be immediately expressed as the intersection of multiple half-planes as in Problem \ref{plm:extended}. Instead, we express the safe set as the union of multiple convex ``clusters" with a separate safety filter for each cluster. At each time $t$, then, $\vec{u}(t)$ is chosen as the closest solution to $\vec{u}_{nom}(t)$ among the solutions to all safety filters for which the system is currently safe. If each of the safety filters satisfies Conjecture \ref{cnj:extended_solution} for its respective cluster, then the closed-loop system will remain safe, as $\vec{u}(t)$ will always ensure forward-invariance of at least one cluster and thus forward-invariance of the overall safe set.

\subsection{Control and Safety Filter Design}

We first design a nominal tracking controller to take the drone from its initial position to a target location as follows. We first design a stabilizing feedback gain, $K_{lqr}$, using LQR on the dynamics $(A, B_1B_2)$. Then, consider the output vector
\begin{equation}
    \vec{y} = [x, y, z, \psi]^\top = C\vec{x}.
\end{equation}
In order to track a desired setpoint $\vec{y}_d$, we apply the nominal control input
\begin{equation}
    \vec{u}_{nom}(t) = K_{lqr}\vec{x}(t) + T_{DC}^{-1}\vec{y}_d
\end{equation}
where $T_{DC}$ is the DC gain of the closed-loop transfer function matrix
\begin{equation}
    T(s) = C(sI - A - B_1B_2K_{lqr})^{-1}B_1B_2.
\end{equation}

The quadrotor model in Section \ref{subsec:quadrotor_modeling} is a 4th-order integrator in $x$ and $y$ and a 2nd-order integrator in $z$ and $\psi$, and the safety constraints are only in $x$ and $y$. Thus, to manipulate the plant into the form required for Problem \ref{plm:extended}, we consider only the dynamics in the $x$-$y$ plane. Defining $\vec{x}_1 = [x, y]^\top$, $\vec{x}_2 = [v_x, v_y]^\top$, $\vec{x}_3 = [g\theta, -g\phi]^\top$, and $\vec{x}_4 = [gq, -gp]^\top$, the dynamics of the constrained states are now given by
\begin{equation}
    \dot{\vec{x}}_1 = \vec{x}_2,\; \dot{\vec{x}}_2 = \vec{x}_3,\; \dot{\vec{x}}_3 = \vec{x}_4,\; \dot{\vec{x}}_4 = B_3\vec{u}
\end{equation}
where $B_3 \in \bbR^{2 \times 4}$ can be derived from \eqref{eqn:quadrotor_forces_thrusts}.

The safe set consists of two clusters with three state constraints each. We design one safety filter for each cluster as follows. Each cluster consists of constraints of the form in \eqref{eqn:md_state_constraints} with $m = 2$ and $p = 3$. For each $k$, functions $h_{k1}, \dots, h_{k4}$ and $\Delta_{k2}, \Delta_{k3}$ are defined as in \eqref{eqn:extended_h_1}-\eqref{eqn:extended_Delta_i}, and the parameters are chosen as $\gamma_{ki} = \gamma_i\ \forall k$, $\epsilon_{ki} = \epsilon_i\ \forall k$ for simplicity. The safety filter consists of \eqref{eqn:safety_filter_cost}-\eqref{eqn:safety_filter_input_constraint} with $N = 3$ and, for each $k \in \{1, 2, 3\}$,
\begin{subequations}
    \begin{align}
        \vec{b}_k =&\ B_3^\top\vec{r}_k, \\
        c_k =&\ \gamma_{4}\sqrt{h_{k4}} - \frac{\gamma_{3}^2}{2} - \epsilon_{4} + \frac{\gamma_{3}}{2\sqrt{h_{k3}}}(h_{k4} + \Delta_{k3} + \epsilon_{3}).
    \end{align}
\end{subequations}
Since only two of the three state constraint functions $h_{k1}$ can ever go to zero simultaneously, the parameters $\gamma_i, \epsilon_i$ were tuned by making the constraints pairwise compatible: for each pair $(k, \ell)$, the parameters were tuned such that
\begin{equation}
    \frac{\gamma_{3}^2}{2} + \epsilon_4 \leq \min\{\bar{u}\hat{u}^\top B_3^\top\vec{r}_k, \bar{u}\hat{u}^\top B_3^\top\vec{r}_\ell\}
\end{equation}
for a $\hat{u} \in \bbR^2, \|\hat{u}\| = 1$ such that $\hat{u}^\top B_3^\top\vec{r}_k > 0$ and $\hat{u}^\top B_3^\top\vec{r}_\ell > 0$. This tuning ensured that $\bar{u}\hat{u}$ would be a feasible input in the event that either $h_{k1}$, $h_{\ell1}$, or both went to zero.

\subsection{Results}

Figures \ref{fig:traj_1_top_view}-\ref{fig:traj_1_thrusts} show the results of a simulation where the start and goal locations are in the safe set, but the straight path connecting them is not. As expected, the quadrotor tracks the goal as well as it can while respecting its input limitations and avoiding the walls of the room. Appendix \ref{app:simulations} presents additional simulation results, including simulations with severe input limits.

% Note that there is some conservatism apparent, where the quadrotor keeps some distance from the wall and never quite saturates its inputs. This is because the simulations were done in Simulink with a zero-order-hold controller at 100 Hz. The safety filter parameters needed to be tuned conservatively in order to ensure safety under the zero-order hold.

\begin{figure}
    \centering
    \begin{subfigure}[t]{0.48\textwidth}
        \centering
        \includegraphics[width=\textwidth]{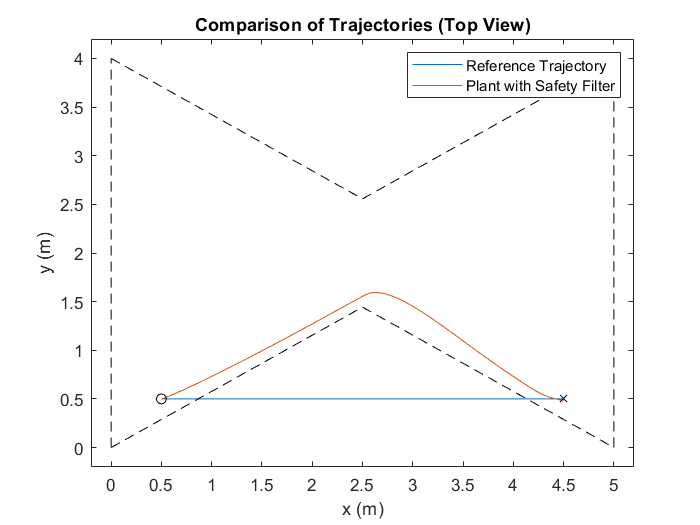}
        \caption{Top-down view of the plant's trajectory under the safety filter. The plant reaches the target while staying inside the safe set (dashed black lines).}
        \label{fig:traj_1_top_view}
    \end{subfigure}
    \hfill
    % \begin{subfigure}[t]{0.3\textwidth}
    %     \centering
    %     \includegraphics[width=\textwidth]{Images/x_y_pos_large_input_traj_1.png}
    %     \caption{$x(t)$ and $y(t)$ of the plant under the safety filter.}
    %     \label{fig:traj_1_x_y_pos}
    % \end{subfigure}
    % \hfill
    \begin{subfigure}[t]{0.48\textwidth}
        \centering
        \includegraphics[width=\textwidth]{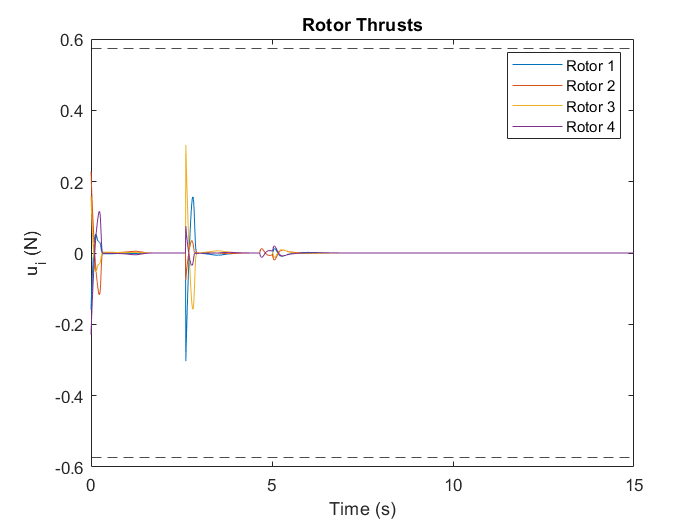}
        \caption{Individual rotor thrusts resulting from the safety filter. All rotors obey the input limits (dashed black lines).}
        \label{fig:traj_1_thrusts}
    \end{subfigure}
\end{figure}
\section{Conclusions and Future Work} \label{sec:conclusions_future_work}

This work proposed a new method of synthesizing analytical safety filters for multi-integrator plants of any relative degree under input saturation and multiple simultaneous state constraints. Explicit methods were given for both building the safety filter and tuning it for implementability in the single-input setting. Additionally, an extension to the multi-input setting was discussed, including special cases under which the single-input approach can be straightforwardly applied and a conjecture that it can be applied in the more general extension as well.

While the method proposed in this work is powerful in its simplicity - the safety filter consists of a collection of analytical constraints on the control input with no online optimization or explicit set construction required - the assumption of a multi-integrator plant limits its applicability. Thus, future work may revolve around expanding the scope of this approach. One possible expansion could be allowing the input constraints to vary with the state, as this may enable application to arbitrary plants under nonlinear dynamic inversion.

\bibliographystyle{IEEEtran}
\bibliography{references}

\appendix
\section{Proof of Lemma \ref{lem:key_technical_lemma}} \label{app:technical_lemma_proof}

Suppose that $\bar{S}$ is not rendered forward-invariant. Then, $\vec{x}(0) \in \bar{S}$ and $\exists t_1 > 0$ such that $\vec{x}(t_1) \notin \bar{S}$, which is equivalent to saying that there exist $\ell_1, \dots, \ell_q$ such that $h_{\ell_i}(\vec{x}(t_1)) < 0\ \forall i \in \{1, \dots, q\}$ for some $1 \leq q \leq p$. As $\vec{x}(0) \in \bar{S}$, we additionally know that $h_{\ell_i}(\vec{x}(0)) \geq 0\ \forall i \in \{1, \dots, q\}$. Now, from \eqref{eqn:prelim_plant}, we know that $\vec{x}(t)$ is continuous and differentiable, and by assumption, $h_{\ell_i}(\vec{x})$ is continuous and differentiable. Thus, $h_{\ell_i}(t)$ is continuous and differentiable. Therefore, if $h_{\ell_i}(0) \geq 0$ and $h_{\ell_i}(t_1) < 0$, then there must exist $\bar{t}_1, \dots, \bar{t}_q \in [0, t_1)$ such that $h_{\ell_i}(t) \geq 0\ \forall t < \bar{t}_i$, $h_{\ell_i}(\bar{t}_i) = 0$, and $\dot{h}_{\ell_i}(\bar{t}_i) < 0$ (i.e. $\bar{t}_i$ is the first time at which $h_{\ell_i}$ becomes negative). In particular, consider $\bar{t} = \min_i \bar{t}_i$. Then, $h_{\ell_i}(\bar{t}) \geq 0\ \forall i$ and there exists at least one $\ell_i$ such that $h_{\ell_i}(\bar{t}) = 0$ and $\dot{h}_{\ell_i}(\bar{t}) < 0$.

However, $h_{\ell_i}(\bar{t}) \geq 0\ \forall i \implies \vec{x}(\bar{t}) \in \bar{S}$. From Definitions \ref{def:safety_filter} and \ref{def:implementability}, if \eqref{eqn:safety_filter_cost}-\eqref{eqn:safety_filter_input_constraint} is an implementable safety filter for the input-constrained plant in \eqref{eqn:prelim_plant}-\eqref{eqn:input_constraint} with respect to $\bar{S}$, then we know that $\dot{h}_\ell(\bar{t}) \geq -\alpha_\ell(h_\ell(\bar{t}))\ \forall \ell$. In particular, for any $\ell_i$ such that $h_{\ell_i}(\bar{t}) = 0$, we have $\dot{h}_{\ell_i} \geq 0$. Thus, we reach a contradiction. \hfill $\Box$
\section{Derivation of CBF Recursion for Problem \ref{plm:simplified}} \label{app:ICCBF_derivation}

This appendix provides a derivation for \eqref{eqn:simplified_h_1}-\eqref{eqn:simplified_Delta_i}, and also serves to connect our approach to that of \cite{agrawal2021ICCBFs}.

The approach in \cite{agrawal2021ICCBFs} applies backstepping as in the usual approach to high-order CBFs (see e.g. \cite{xiao2019HOCBF}), but seeks a particular series of class $\calK_\infty$ functions $\alpha_1, \dots, \alpha_N$ ($N \geq n$) such that $h_N$ is an ICCBF (see Definition \ref{def:ICCBF}. However, it is unclear in \cite{agrawal2021ICCBFs} how to choose $\alpha_1, \dots, \alpha_N$ such that $N$ is finite. Our approach adds two innovations to make the problem analytically tractable and allow $N$ to be finite:
\begin{enumerate}
    \item Rather than searching over the full space of $\calK_\infty$ functions, we seek functions of the particular form $\alpha_i(h_i) = \gamma_ih_i^{k_i}$ for some $\gamma_i, k_i > 0$.
    \item We introduce a nonnegative offset $\Delta_i(t) \geq 0$ into the inequality for each $h_i$, which will prevent us from needing to carry a polynomially-growing number of terms all the way through the derivation.
\end{enumerate}
In summary, for each CBF $h_i(\vec{x})$, we will seek to ensure that
\begin{equation} \label{eqn:ICCBF_safety_filter_archetype}
    \dot{h}_i \geq -\gamma_ih_i^{k_i} + \Delta_i,
\end{equation}
which, if satisfied with $\gamma_i, k_i > 0$ and $\Delta_i(t) \geq 0\ \forall t \geq 0$, will guarantee that $h_i(t) \geq 0\ \forall t \geq 0$.

Consider the 1st-order CBF candidate $h_1(\vec{x})$ as in \eqref{eqn:simplified_h_1}, which is chosen such that
\begin{equation}
    S_1 := S = \{\vec{x} \in \bbR^n : h_1(\vec{x}) \geq 0\}
\end{equation}
for $S$ in \eqref{eqn:state_constraints_simple}. In order to render $S$ forward-invariant, we seek to ensure
\begin{equation} \label{eqn:h1_filter_1}
    \dot{h}_1 \geq -\gamma_1h_1^{k_1} + \Delta_1
\end{equation}
for some $\gamma_1, k_1 > 0$ and $\Delta_1(t) \geq 0$. For simplicity, we choose $\Delta_1 = 0$. If $n = 1$, then $\dot{h}_1 = u$ and \eqref{eqn:h1_filter_1} becomes
\begin{equation} \label{eqn:h1_filter_2}
    u \geq -\gamma_1h_1^{k_1}.
\end{equation}
If a safety filter with \eqref{eqn:h1_filter_2} is applied to $u$, it will ensure that $h_1(t) \geq 0\ \forall t \geq 0$, which in turn will ensure that the right-hand side of the inequality is always nonpositive and thus that $u = 0$ is always a valid solution to the filter. Thus, for the $n = 1$ case, $h_1$ is always an ICCBF, and a safety filter with \eqref{eqn:h1_filter_2} is implementable.

Now suppose that $n > 1$. Then, \eqref{eqn:h1_filter_1} is equivalent to
\begin{equation} \label{eqn:h1_filter_3}
    x_2 + \gamma_1h_1^{k_1} \geq 0.
\end{equation}
We thus define the second-order CBF candidate as
\begin{subequations}
    \begin{align}
        h_2(\vec{x}) &= \dot{h}_1 + \gamma_1h_1^{k_1} - \Delta_1 \label{eqn:h2_definition} \\
        &= x_2 + \gamma_1h_1^{k_1} \label{eqn:h2_rewritten}
    \end{align}
\end{subequations}
with corresponding 2nd-order safe set $S_2 = \{\vec{x} \in \bbR^n : h_2(\vec{x}) \geq 0\}$. If $h_2(\vec{x}(0)) \geq 0$ and $S_2$ is rendered forward-invariant, then by definition, \eqref{eqn:h1_filter_1} is satisfied for all $t \geq 0$, rendering $S_1$ forward-invariant. We seek to render $S_2$ forward-invariant by satisfying the inequality
\begin{equation} \label{eqn:h2_filter_1}
    \dot{h}_2 \geq -\gamma_2h_2^{k_2} + \Delta_2.
\end{equation}
for some $\gamma_2, k_2 > 0$ and $\Delta_2(t) \geq 0$. Differentiating \eqref{eqn:h2_rewritten} and using \eqref{eqn:h2_definition} to substitute in for $\dot{h}_1$, we can rewrite \eqref{eqn:h2_filter_1} as
\begin{equation} \label{eqn:h2_filter_2}
    \dot{x}_2 + \underbrace{\gamma_1k_1h_1^{k_1-1}(h_2 + \Delta_1)}_{c_1} - \underbrace{\gamma_1^2k_1h_1^{2k_1-1}}_{c_2} \geq -\underbrace{\gamma_2h_2^{k_2}}_{c_3} + \Delta_2.
\end{equation}
If this inequality is satisfied for all $t \geq 0$, it will ensure $h_1(t), h_2(t) \geq 0\ \forall t \geq 0$. Crucially, if this occurs, $c_1$, $c_2$, and $c_3$ are all sign-definite, and furthermore their signs are such that $c_1$ and $c_3$ strictly help in satisfying the inequality and $c_2$ strictly hurts in satisfying the inequality.
% As $c_1$ strictly helps, to avoid carrying the terms through the rest of the derivation, we can choose $\Delta_2$ to cancel it out:
% \begin{equation} \label{eqn:Delta2_definition}
%     \Delta_2 = \gamma_1k_1h_1^{k_1-1}(h_2 + \Delta_1).
% \end{equation}
% On the other hand, $c_2$ is strictly harmful and, depending on the value of $k_1$, may be arbitrarily large if $h_1$ goes to zero or infinity. In order to prevent it from becoming arbitrarily large, we will choose $k_1 = \frac{1}{2}$ to make $c_2$ constant. Finally, \eqref{eqn:b2_filter_1} is equivalent to
% \begin{equation} \label{eqn:b2_filter_3}
%     \dot{x}_2 + \gamma_2h_2^{k_2} + \frac{\gamma_1^2}{2} \geq 0
% \end{equation}
% with
% \begin{equation} \label{eqn:Delta2_rewritten}
%     \Delta_2 = \frac{\gamma_1}{2\sqrt{h_1}}(h_2 + \Delta_1).
% \end{equation}
As $c_2$ is strictly harmful, we want to prevent it from becoming arbitrarily large if $h_1$ goes to zero or infinity. We can do this by choosing $k_1 = \frac{1}{2}$ so that $c_2 = \frac{\gamma_1^2}{2}$.

If $n = 2$, then $\dot{x}_2 = u$ and, setting $\Delta_2 = 0$, \eqref{eqn:h2_filter_2} is equivalent to
\begin{equation} \label{eqn:h2_filter_3}
    u \geq -\gamma_2h_2^{k_2} - \frac{\gamma_1}{2\sqrt{h_1}}(h_2 + \Delta_1) + \frac{\gamma_1^2}{2}.
\end{equation}
If this inequality is satisfied for all $t \geq 0$, then all three terms on the left-hand side are sign-definite, and the only harmful term is the constant $\frac{\gamma_1^2}{2}$. Furthermore, we have free choice of $\gamma_1 > 0$. Therefore, by choosing any $\gamma_1 \leq \sqrt{2\bar{u}}$, we can ensure that $u = \bar{u}$ always satisfies the inequality. Thus, for the $n = 2$ case with $k_1 = \frac{1}{2}$, $\gamma_1 \leq \sqrt{2\bar{u}}$, and $\Delta_1 = \Delta_2 = 0$, $h_2$ is an ICCBF, and a safety filter with \eqref{eqn:h2_filter_3} is implementable.

We will derive one more level before extrapolating to the full $n$th-order case. Suppose now that $n > 2$. As $c_1$ in \eqref{eqn:h2_filter_2} is strictly nonnegative, we can choose
\begin{equation} \label{eqn:Delta2}
    \Delta_2 = c_1 = \frac{\gamma_1}{2\sqrt{h_1}}(h_2 + \Delta_1)
\end{equation}
to cancel it out and avoid having to keep track of a growing number of terms. Then, \eqref{eqn:h2_filter_1} and \eqref{eqn:h2_filter_2} are equivalent to
\begin{equation} \label{eqn:h2_filter_4}
    x_3 + \gamma_2h_2^{k_2} - \frac{\gamma_1^2}{2} \geq 0.
\end{equation}
We thus define the third-order CBF candidate as
\begin{subequations}
    \begin{align}
        h_3(\vec{x}) &= \dot{h}_2 + \gamma_2h_2^{k_2} - \Delta_2 \label{eqn:h3_definition} \\
        &= x_3 + \gamma_2h_2^{k_2} - \frac{\gamma_1^2}{2} \label{eqn:h3_rewritten}
    \end{align}
\end{subequations}
with corresponding 3rd-order safe set $S_3 = \{\vec{x} \in \bbR^n : h_3(\vec{x}) \geq 0\}$. If $h_2(\vec{x}(0)), h_3(\vec{x}(0)) \geq 0$ and $S_3$ is rendered forward-invariant, then by definition, \eqref{eqn:h2_filter_1} is satisfied, rendering $S_2$ and thus $S_1$ forward-invariant. We seek to render $S_3$ forward-invariant by satisfying the inequality
\begin{equation} \label{eqn:h3_filter_1}
    \dot{h}_3 \geq -\gamma_3h_3^{k_3} + \Delta_3
\end{equation}
for some $\gamma_2, k_2 > 0$ and $\Delta_2(t) \geq 0$. Differentiating \eqref{eqn:h3_rewritten} and using \eqref{eqn:h3_definition} to substitute in for $\dot{h}_2$, we can rewrite \eqref{eqn:h3_filter_1} as
\begin{equation} \label{eqn:h3_filter_2}
    \dot{x}_3 + \underbrace{\gamma_2k_2h_2^{k_2-1}(h_3 + \Delta_2)}_{d_1} - \underbrace{\gamma_2^2k_2h_2^{2k_2-1}}_{d_2} \geq -\underbrace{\gamma_3h_3^{k_3}}_{d_3} + \Delta_3.
\end{equation}
If this inequality is satisfied for all $t \geq 0$, it will ensure $h_1(t), h_2(t), h_3(t) \geq 0\ \forall t \geq 0$. Crucially, if this occurs, $d_1$, $d_2$, and $d_3$ are all sign-definite, and furthermore their signs are such that $d_1$ and $d_3$ strictly help in satisfying the inequality and $d_2$ strictly hurts in satisfying the inequality. As before, we can ensure that $d_2$ does not become arbitrarily harmful by choosing $k_2 = \frac{1}{2}$ so that $d_2 = \frac{\gamma_2^2}{2}$. If $n = 3$, then $\dot{x}_3 = u$ and, setting $\Delta_3 = 0$, \eqref{eqn:h3_filter_1} is equivalent to
\begin{equation} \label{eqn:h3_filter_2}
    u \geq -\gamma_3h_3^{k_3} - \frac{\gamma_2}{2\sqrt{h_2}}(h_3 + \Delta_2) + \frac{\gamma_2^2}{2}.
\end{equation}
As before, if $\gamma_2 \leq \sqrt{2\bar{u}}$, then $h_3$ is an ICCBF and a safety filter with \eqref{eqn:h3_filter_2} is implementable. Note that in the case of $n = 3$, we choose $k_1 = k_2 = \frac{1}{2}$, $0 < \gamma_2 \leq \sqrt{2\bar{u}}$, $\Delta_1 = \Delta_3 = 0$, and $\Delta_2$ according to \eqref{eqn:Delta2}, but we do not need any upper bound on $\gamma_1 > 0$. (This will no longer be true when we introduce multiple simultaneous constraints.)

We now extrapolate to the full $n$th-order case. Choose $k_i = \frac{1}{2}\ \forall i$. Then, defining $\gamma_0 = 0$, the $i$th-order CBF is given recursively by
\begin{subequations}
    \begin{align}
        h_1(\vec{x}) &= x_1 - \ubar{x}_1, \label{eqn:h1_definition_2} \\
        h_i(\vec{x}) &= \dot{h}_{i-1} + \gamma_{i-1}\sqrt{h_{i-1}} - \Delta_{i-1} \label{eqn:hi_definition} \\
        &= x_i + \gamma_{i-1}\sqrt{h_{i-1}} - \frac{\gamma_{i-2}^2}{2}\ \forall i \geq 2 \label{eqn:hi_rewritten}
    \end{align}
\end{subequations}
with corresponding $i$th-order safe set $S_i = \{\vec{x} \in \bbR^n : h_i(\vec{x}) \geq 0\}$, where \eqref{eqn:hi_definition} is the definition of $h_i(\vec{x})$ and \eqref{eqn:hi_rewritten} is given by choosing $\Delta_i$ recursively as
\begin{equation} \label{eqn:Deltai}
    \Delta_1 = 0, \; \Delta_i = \frac{\gamma_{i-1}}{2\sqrt{h_{i-1}}}(h_i + \Delta_{i-1})\ \forall i \geq 2.
\end{equation}
In order to ensure forward-invariance of $S_i$, we seek to ensure that
\begin{equation} \label{eqn:hi_filter_1}
    \dot{h}_i \geq -\gamma_i\sqrt{h_i} + \Delta_i.
\end{equation}

First, consider $i < n$. Differentiating \eqref{eqn:hi_rewritten} and using \eqref{eqn:hi_definition} to substitute in for $\dot{h}_i$, and using $\Delta_i$ given in \eqref{eqn:Deltai}, we can show that \eqref{eqn:hi_filter_1} is equivalent to
\begin{equation} \label{eqn:hi_filter_2}
    x_{i+1} + \gamma_i\sqrt{h_i} - \frac{\gamma_{i-1}^2}{2} \geq 0
\end{equation}
which will indeed result in the $(i+1)$th-order CBF
\begin{subequations}
    \begin{align}
        h_{i+1}(\vec{x}) &= \dot{h}_i + \gamma_i\sqrt{h_i} - \Delta_i \label{eqn:bip1_definition} \\
        &= x_{i+1} + \gamma_i\sqrt{h_i} - \frac{\gamma_{i-1}^2}{2}, \label{eqn:bip1_rewritten}
    \end{align}
\end{subequations}
proving the recursion in \eqref{eqn:h1_definition_2}-\eqref{eqn:hi_rewritten}.

Now consider $i = n$. Differentiating \eqref{eqn:hi_rewritten} and using \eqref{eqn:hi_definition} to substitute in for $\dot{h}_{n-1}$ and choosing $\Delta_n = 0$, we can show that \eqref{eqn:hi_filter_1} is equivalent to
\begin{equation} \label{eqn:u_filter}
    u \geq -\gamma_n\sqrt{h_n} - \frac{\gamma_{n-1}}{2\sqrt{h_{n-1}}}(h_n + \Delta_{n-1}) + \frac{\gamma_{n-1}^2}{2}
\end{equation}
which will result in an implementable safety filter if $\gamma_{n-1} \leq \sqrt{2\bar{u}}$.

\subsection{An Important Modification for Implementation}

If the inequality in \eqref{eqn:u_filter} is satisfied for all $t \geq 0$, it will ensure that $h_1(t), \cdots, h_n(t) \geq 0\ \forall t \geq 0$. However, this does not preclude $h_i(t) \to 0$ for any $i$, and the recursion requires division by $\sqrt{h_i}$ for every $i < n$. Thus, in order to prevent division by infinitesimal numbers, it is important to modify the filter in \eqref{eqn:hi_filter_1} with a small hyperparameter $\epsilon_i > 0$, so that we now seek to ensure that
\begin{equation} \label{eqn:hi_filter_1_with_epsilon}
    \dot{h}_i \geq -\gamma_i\sqrt{h_i} + \Delta_i + \epsilon_i.
\end{equation}
It is straightforward to verify that the recursion in \eqref{eqn:simplified_h_1}-\eqref{eqn:simplified_Delta_i} is correct to incorporate $\epsilon_i$.
\section{Proof of Theorem \ref{thm:simplified_solution}} \label{app:simplified_proof}

There are two steps to proving Theorem \ref{thm:simplified_solution}:
\begin{enumerate}
    \item Show that \eqref{eqn:safety_filter_cost}-\eqref{eqn:safety_filter_input_constraint} with $N = 1$, $b_1 = 1$, and $c_1$ as in \eqref{eqn:simplified_c} satisfies the definition of a safety filter in Definition \ref{def:safety_filter} with respect to $\bar{S}$ as in \eqref{eqn:simplified_composite_safe_set}.
    \item Show that the safety filter with $\gamma_i, \epsilon_i > 0$ and $\gamma_{n-1} \leq \sqrt{2\bar{u}}$ satisfies the definition of implementability in Definition \ref{def:implementability}.
\end{enumerate}
These two steps complete the proof. For the first step, we will find the following lemma useful:
\begin{lemma} \label{lem:simplified_h_i_dot}
    Consider $h_i(\vec{x})$ and $\Delta_i(\vec{x})$ as defined in \eqref{eqn:simplified_h_1}-\eqref{eqn:simplified_Delta_i}. Then, for every $i \in \{1, \dots, n-1\}$, we have
    \begin{equation}
        \dot{h}_i = x_{i+1} + \Delta_i - \frac{\gamma_{i-1}^2}{2}\ \forall i \in \{2, \dots, n-1\}
    \end{equation}
\end{lemma}
\begin{proof}
    From \eqref{eqn:simplified_h_1}, we have
    \begin{equation}
        \dot{h}_1 = x_2 = x_2 + \Delta_1 - \frac{\gamma_0^2}{2}.
    \end{equation}
    Now, suppose that for any $i \in \{2, \dots, n-1\}$, we have $\dot{h}_{i-1} = x_i + \Delta_{i-1} - \frac{\gamma_{i-2}^2}{2}$. Then,
    \begin{align}
        \dot{h}_i &= x_{i+1} + \frac{\gamma_{i-1}}{2\sqrt{h_{i-1}}}\dot{h}_{i-1} \nonumber \\
        &= x_{i+1} + \frac{\gamma_{i-1}}{2\sqrt{h_{i-1}}}(x_i - \frac{\gamma_{i-2}^2}{2} + \Delta_{i-1}) \nonumber \\
        &= x_{i+1} + \frac{\gamma_{i-1}}{2\sqrt{h_{i-1}}}(h_i - \gamma_{i-1}\sqrt{h_{i-1}} + \epsilon_{i-1} + \Delta_{i-1}) \nonumber \\
        &= x_{i+1} + \Delta_i - \frac{\gamma_{i-1}^2}{2}
    \end{align}
    proving the recursion.
\end{proof}

\subsubsection*{Step 1: Proving the Safety Filter}

Define $g_1(\vec{x}), \dots, g_n(\vec{x})$ as
\begin{equation} \label{eqn:thm_1_gi_definition}
    g_i(\vec{x}) := h_i(\vec{x}) - \frac{\epsilon_i^2}{\gamma_i^2}\ \forall i \in \{1, \dots, n\}
\end{equation}
for $h_i(\vec{x})$ as in \eqref{eqn:simplified_h_1}-\eqref{eqn:simplified_h_i}. Then, \eqref{eqn:simplified_composite_safe_set} is equivalent to
\begin{equation}
    \bar{S} = \{\vec{x} \in \bbR^n : g_i(\vec{x}) \geq 0\ \forall i = 1, \dots, n\}
\end{equation}
Thus, under Definition \ref{def:safety_filter}, Theorem \ref{thm:simplified_solution} describes a safety filter if, whenever $\vec{x} \in \bar{S}$, $b_1u + c_1 \geq 0$, and $\ubar{u} \leq u \leq \bar{u}$, we have $\dot{g}_i \geq -\alpha_i(g_i)\ \forall i$ for some $\alpha_i \in \calK_\infty$.

First, we will show that whenever $\vec{x} \in \bar{S}$, we have $\dot{g}_i(\vec{x}) \geq -\gamma_i\sqrt{g_i(\vec{x})}\ \forall i \in \{1, \dots, n-1\}$ regardless of $u$, simply by the definitions of $\bar{S}$, $h_i(\vec{x})$, and $\Delta_i(\vec{x})$. We will then proceed to show that choosing $u$ according to $b_1u + c_1 \geq 0$ ensures $\dot{g}_n(\vec{x}) \geq -\gamma_n\sqrt{g_n(\vec{x})}$, completing the first step of the proof.

For the first part, $\vec{x} \in \bar{S} \implies h_i(\vec{x}) \geq \frac{\epsilon_i^2}{\gamma_i^2} > 0\ \forall i$, which in turn implies that $\Delta_i \geq 0\ \forall i$. Then, using Lemma \ref{lem:simplified_h_i_dot}, for any $i \in \{1, \dots, n-1\}$, we have
\begin{align}
    \dot{g}_i &= \dot{h}_i = x_{i+1} - \frac{\gamma_{i-1}^2}{2} + \Delta_i \nonumber \\
    &= h_{i+1} - \gamma_i\sqrt{h_i} + \epsilon_i + \Delta_i \nonumber \\
    &> -\gamma_i\sqrt{h_i} + \epsilon_i \nonumber \\
    &= -\gamma_i\sqrt{g_i + \frac{\epsilon_i^2}{\gamma_i^2}} + \epsilon_i \nonumber \\
    &\geq -\gamma_i\sqrt{g_i} \label{eqn:simplified_g_i_dot}
\end{align}
since $\sqrt{a + b} \leq \sqrt{a} + \sqrt{b}$ for any $a, b \geq 0$.

For the second part, from Lemma \ref{lem:simplified_h_i_dot}, we have
\begin{align}
    \dot{g}_n &= \dot{h}_n = u + \frac{\gamma_{n-1}}{2\sqrt{h_{n-1}}}\dot{h}_{n-1} \nonumber \\
    &= u + \frac{\gamma_{n-1}}{2\sqrt{h_{n-1}}}(x_n - \frac{\gamma_{n-2}^2}{2} + \Delta_{n-1}) \nonumber \\
    &= u + \frac{\gamma_{n-1}}{2\sqrt{h_{n-1}}}(h_n - \gamma_{n-1}\sqrt{h_{n-1}} + \epsilon_{n-1} + \Delta_{n-1}) \nonumber \\
    &= u + \frac{\gamma_{n-1}}{2\sqrt{h_{n-1}}}(h_n + \epsilon_{n-1} + \Delta_{n-1}) - \frac{\gamma_{n-1}^2}{2}.
\end{align}
Thus, choosing $u$ such that $b_1u + c_1 \geq 0$ for $b_1 = 1$ and $c_1$ as in \eqref{eqn:simplified_c} ensures that
\begin{align}
    \dot{g}_n &\geq -\gamma_n\sqrt{h_n} + \epsilon_n \nonumber \\
    &\geq -\gamma_n\sqrt{g_n}
\end{align}
similarly to \eqref{eqn:simplified_g_i_dot}.

Therefore, the filter described in Theorem \ref{thm:simplified_solution} is a valid safety filter.

\subsubsection*{Step 2: Proving Implementability}

The safety filter described in Theorem \ref{thm:simplified_solution} is implementable if, whenever $\vec{x} \in \bar{S}$, there exists a $u \in [\ubar{u}, \bar{u}]$ satisfying $b_1u + c_1 \geq 0$. Whenever $\vec{x} \in \bar{S}$, we have $h_i(\vec{x}) \geq \frac{\epsilon_i^2}{\gamma_i^2}$. Thus, we have
\begin{align}
    b_1u + c_1 &\geq u + \gamma_n\sqrt{h_n} - \frac{\gamma_{n-1}^2}{2} - \epsilon_n \nonumber \\
    &\geq u - \frac{\gamma_{n-1}^2}{2}
\end{align}
and thus $b_1u + c_1 \geq 0 \impliedby u \geq \frac{\gamma_{n-1}^2}{2}$. The choice of $\gamma_{n-1} \leq \sqrt{2\bar{u}}$ therefore ensures that $u = \bar{u}$ is always a feasible solution to the optimization whenever $\vec{x} \in \bar{S}$, and thus that the safety filter is implementable. \hfill $\Box$
\section{Algorithm for Parameter Tuning} \label{app:parameter_tuning}

\begin{algorithm}[H]
    \caption{Parameter Tuning}
    \label{alg:parameter_tuning}
    \begin{algorithmic}
        \STATE{{\bf inputs:} $\gamma_{1,nom} > 0$, $\delta$, $\beta_1, \dots, \beta_n$, $\alpha_2, \dots, \alpha_n$, $\eta_1, \dots, \eta_{n-1}$, $\varepsilon > 0$}
        \IF{n == 1}
            \STATE{$\gamma_1 \gets \gamma_{1,nom}$}
            \STATE{Calculate $\epsilon_1$ from \eqref{eqn:epsilon_1_upper_limit} using $\gamma_1$, replacing inequality with equals}
            \STATE{{\bf return} $\gamma_1, \epsilon_1$}
        \ENDIF
        \STATE \COMMENT{Find an impermissibly large $\gamma_1$}
        \STATE{$\gamma_{1,u} \gets \gamma_{1,nom}$}
        \WHILE{\TRUE}
            \STATE{Calculate $\gamma_2, \dots, \gamma_n$ from \eqref{eqn:gamma_2_lower_limit}-\eqref{eqn:gamma_i_lower_limit} using $\gamma_{1,u}$, replacing inequalities with equals}
            \IF{any $\gamma_i$ does not satisfy \eqref{eqn:gamma_i_upper_limit}-\eqref{eqn:gamma_n-1_upper_limit}}
                \STATE{{\bf break}}
            \ENDIF
            \STATE{$\gamma_{1,u} \gets 2\gamma_{1,u}$}
        \ENDWHILE
        \STATE \COMMENT{Find a permissibly small $\gamma_1$}
        \STATE{$\gamma_{1,\ell} \gets \gamma_{1,nom}$}
        \WHILE{\TRUE}
            \STATE{Calculate $\gamma_2, \dots, \gamma_n$ from \eqref{eqn:gamma_2_lower_limit}-\eqref{eqn:gamma_i_lower_limit} using $\gamma_{1,\ell}$, replacing inequalities with equals}
            \IF{all $\gamma_i$ satisfy \eqref{eqn:gamma_i_upper_limit}-\eqref{eqn:gamma_n-1_upper_limit}}
                \STATE{{\bf break}}
            \ENDIF
            \STATE{$\gamma_{1,\ell} \gets \gamma_{1,\ell}/2$}
        \ENDWHILE
        \STATE \COMMENT{Find the largest permissible value of $\gamma_1$ using binary search}
        \WHILE{$\gamma_{1,u} - \gamma_{1,\ell} > \varepsilon$}
            \STATE{$\gamma_1 \gets (\gamma_{1,u} + \gamma_{1,\ell})/2$}
            \STATE{Calculate $\gamma_2, \dots, \gamma_n$ from \eqref{eqn:gamma_2_lower_limit}-\eqref{eqn:gamma_i_lower_limit} using $\gamma_1$, replacing inequalities with equals}
            \IF{all $\gamma_i$ satisfy \eqref{eqn:gamma_i_upper_limit}-\eqref{eqn:gamma_n-1_upper_limit}}
                \STATE{$\gamma_{1,\ell} \gets \gamma_1$}
            \ELSE
                \STATE{$\gamma_{1,u} \gets \gamma_1$}
            \ENDIF
        \ENDWHILE
        \STATE{$\gamma_1 \gets \gamma_{1,\ell}$}
        \STATE{Calculate $\gamma_2, \dots, \gamma_n, \epsilon_1, \dots, \epsilon_n$ from \eqref{eqn:gamma_2_lower_limit}-\eqref{eqn:epsilon_i_upper_limit} using $\gamma_1$, replacing inequalities with equals}
        \STATE{{\bf return} $\gamma_1, \dots, \gamma_n, \epsilon_1, \dots, \epsilon_n$}
    \end{algorithmic}
\end{algorithm}

\section{Proof of Lemma \ref{lem:tuning}} \label{app:tuning_proof}

Algorithm \ref{alg:parameter_tuning} makes use of the fact that, given values for $\gamma_1$, $\delta$, $\alpha_i$, $\beta_i$, and $\eta_i$, one can calculate values for $\gamma_i\ \forall i \geq 2$ and $\epsilon_i\ \forall i$ from \eqref{eqn:gamma_2_lower_limit}-\eqref{eqn:epsilon_i_upper_limit} by setting the inequalities to equations. Furthermore, the resulting values of $\gamma_i$ and $\epsilon_i$ increase with the chosen value of $\gamma_1$ and go to zero in the limit as $\gamma_1$ goes to zero. The inequalities in \eqref{eqn:gamma_i_upper_limit}-\eqref{eqn:gamma_n-1_upper_limit} are then additional restrictions which, if $\gamma_2, \dots, \gamma_n$ are calculated from $\gamma_1$, implicitly represent upper bounds on $\gamma_1$.

A special case is $n = 1$, in which case there turn out to be no limits on $\gamma_1$ save for \eqref{eqn:gamma_1_lower_limit}. In this case, the algorithm calculates $\epsilon_1$ and terminates immediately. Thus, we henceforth focus on the case where $n \geq 2$.

The algorithm operates in three steps:
\begin{enumerate}
    \item Repeatedly double the input $\gamma_{1,nom}$ until a value $\gamma_{1,u}$ is obtained which is too large for all of \eqref{eqn:gamma_1_lower_limit}-\eqref{eqn:gamma_n-1_upper_limit} to be simultaneously satisfied.
    \item Repeatedly half the input $\gamma_{1,nom}$ until a value $\gamma_{1,\ell}$ is obtained which is sufficiently small that all of \eqref{eqn:gamma_1_lower_limit}-\eqref{eqn:gamma_n-1_upper_limit} are simultaneously satisfied.
    \item Search between $\gamma_{1,u}$ and $\gamma_{1,\ell}$ via binary search until convergence to within tolerance $\varepsilon > 0$ to find the largest value of $\gamma_1$ that permits all of \eqref{eqn:gamma_1_lower_limit}-\eqref{eqn:gamma_n-1_upper_limit} to be simultaneously satisfied.
\end{enumerate}

It is trivial that, if finite values $\gamma_{1,u} > 0$ and $\gamma_{1,\ell} > 0$ exist, each of the three steps will terminate in finitely many steps. Thus, the proof is complete if we show that such finite values $\gamma_{1,u} > 0$ and $\gamma_{1,\ell} > 0$ exist.

\subsubsection*{Existence of Too Large $\gamma_1$}

The first step in this proof is to show that there exists a value $\gamma_{1,u}$ such that, if $\gamma_2, \dots, \gamma_n$ are calculated from \eqref{eqn:gamma_2_lower_limit}-\eqref{eqn:gamma_i_lower_limit} setting the inequalities to equations, then at least one of the inequalities in \eqref{eqn:gamma_i_upper_limit}-\eqref{eqn:gamma_n-1_upper_limit} is violated. The subtlety here is that, in the general problem statement, we may not want to require constraints on all state variables, i.e. some of $\ubar{x}_i$ or $\bar{x}_i$ may be $-\infty$ or $+\infty$ respectively. However, we can assume without loss of generality that at least one of $\ubar{x}_1$ and $\bar{x}_1$ is finite, as otherwise we would redefine $x_1$ to be the first state variable that has a constraint. Furthermore, we can assume w.l.o.g that at least one of $\ubar{u}$ and $\bar{u}$ is finite.

If at least one of $\ubar{x}_1$ and $\bar{x}_1$ is finite, then \eqref{eqn:gamma_2_lower_limit}-\eqref{eqn:gamma_i_lower_limit} give lower bounds on $\gamma_2, \dots, \gamma_n$ which scale with $\gamma_1$. Specifically, we must have $\gamma_2 \geq \frac{(1/2 + \alpha_2)\gamma_1^{3/2}}{\sqrt{\beta_1\sqrt{\delta}}} = \Omega(\gamma_1^{3/2})$, $\gamma_3 \geq \frac{(1 + \alpha_3)\gamma_2^2}{\gamma_1\sqrt{\beta_2\alpha_2}} = \Omega(\gamma_1^2)$, $\gamma_4 \geq \frac{(1 + \alpha_4)\gamma_3^2}{\gamma_2\sqrt{\beta_3\alpha_3}} = \Omega(\gamma_1^{5/2})$, $\cdots$, $\gamma_n \geq \frac{(1 + \alpha_n)\gamma_{n-1}^2}{\gamma_{n-2}\sqrt{\beta_{n-1}\alpha_{n-1}}} = \Omega(\gamma_1^{(n+1)/2})$. Furthermore, if at least one of $\ubar{u}$ and $\bar{u}$ is finite, then per \eqref{eqn:gamma_n-1_upper_limit}, we require $\gamma_{n-1} \leq O(\sqrt{\min\{-\ubar{u}, \bar{u}\}})$. Thus, there exists a $\gamma_{1,u} = O(\min\{-\ubar{u}, \bar{u}\}^{1/n})$ which ensures that \eqref{eqn:gamma_2_lower_limit}-\eqref{eqn:gamma_i_lower_limit} and \eqref{eqn:gamma_n-1_upper_limit} cannot be simultaneously satisfied.

\subsubsection*{Existence of Small Enough $\gamma_1$}

Now, accounting for arbitrarily many of the $\ubar{x}_i$ and $\bar{x}_i$ being finite, setting the inequalities in \eqref{eqn:gamma_2_lower_limit}-\eqref{eqn:gamma_i_lower_limit} to equations, one can show that $\gamma_2, \dots, \gamma_n$ scale with $\gamma_1$ as $\gamma_2 = O(\gamma_1^{k_2})$, $\gamma_3 = O(\gamma_1^{k_3})$, $\cdots$, $\gamma_n = O(\gamma_1^{k_n})$ for a sequence $\{k_i\}_{i = 1, 2, \dots, n}$ such that $k_1 = 1$, $k_2 \in \{\frac{3}{2}, 2\}$, and $k_{i-1} < k_i \leq 2^{i-1}\ \forall i \geq 2$. Therefore, defining $\Gamma_1 := \sqrt{2(1 - \eta_1)\min\{-\ubar{x}_3, \bar{x}_3\}}$, $\cdots$, $\Gamma_{n-2} := \sqrt{2(1 - \eta_{n-2})\min\{-\ubar{x}_n, \bar{x}_n\}}$, $\Gamma_{n-1} := \sqrt{2(1 - \eta_{n-1})\min\{-\ubar{u}, \bar{u}\}}$, it is clear that there exists a $\gamma_{1,\ell} = O(\min\{\Gamma_1, \Gamma_2^{1/k_2}, \cdots, \Gamma_{n-1}^{1/k_{n-1}}\})$ which ensures that \eqref{eqn:gamma_2_lower_limit}-\eqref{eqn:gamma_i_lower_limit} and \eqref{eqn:gamma_i_upper_limit}-\eqref{eqn:gamma_n-1_upper_limit} can be simultaneously satisfied. \hfill $\Box$
\section{Proof of Theorem \ref{thm:full_solution}} \label{app:full_proof}

As in the proof of Theorem \ref{thm:simplified_solution}, there are two steps to proving Theorem \ref{thm:full_solution}:
\begin{enumerate}
    \item Show that \eqref{eqn:safety_filter_cost}-\eqref{eqn:safety_filter_input_constraint} with $N = 2n$, and $b_j, c_j, b_{n+j}, c_{n+j}$ as in \eqref{eqn:full_b_lower}-\eqref{eqn:full_c_upper} satisfies the definition of a safety filter in Definition \ref{def:safety_filter} with respect to $\bar{S}$ as in \eqref{eqn:full_composite_safe_set}.
    \item Show that the safety filter with $\gamma_i, \epsilon_i$ chosen by Algorithm \ref{alg:parameter_tuning} satisfies the definition of implementability in Definition \ref{def:implementability}.
    \item Show that $\bar{S}$ in \eqref{eqn:full_composite_safe_set} is non-empty.
\end{enumerate}
The first step proceeds identically to the first step in the proof of Theorem \ref{thm:simplified_solution}, analyzing each $\ubar{S}_{ji}$ and $\bar{S}_{ji}$ in \eqref{eqn:full_safe_subsets_lower}-\eqref{eqn:full_safe_subsets_upper} separately, and is thus omitted for conciseness. We will focus here only on the second and third steps. Note that the third step is a subtle point that was not a concern in the simplified problem setting of Theorem \ref{thm:simplified_solution}, but is of potential concern here.

\subsubsection*{Step 2: Proving Implementability}

Given that \eqref{eqn:safety_filter_cost}-\eqref{eqn:safety_filter_input_constraint} with $N = 2n$ and $b_j$, $c_j$, $b_{n+j}$, and $c_{n+j}$ as in \eqref{eqn:full_b_lower}-\eqref{eqn:full_c_upper} is a safety filter with respect to $\bar{S}$ as in \eqref{eqn:full_composite_safe_set}, implementability requires that the optimization is feasible whenever $\vec{x} \in \bar{S}$. As $u$ is a scalar, the constraints constitute $n+1$ lower bounds and $n+1$ upper bounds on $u$:
\begin{equation} \label{eqn:full_constraints_on_u}
    \max\{\ubar{u}, -c_1, \cdots, -c_n\} \leq u \leq \min\{\bar{u}, c_{n+1}, \cdots, c_{2n}\}.
\end{equation}
Therefore, implementability is equivalent to requiring that whenever $\vec{x} \in \bar{S}$, every upper bound in \eqref{eqn:full_constraints_on_u} is greater than or equal to every lower bound. Taken at face value, this would require analyzing $\frac{n^2 + 3n}{2}$ comparisons (since $\bar{u} > \ubar{u}$ by assumption). However, considering general indices $j, k \in \{1, \dots, n\}$, it turns out that the analysis can be broken down into only a few cases:
\begin{enumerate}
    \item Showing that $c_{n+j} \geq \ubar{u}$, or similarly, that $\bar{u} \geq -c_j$, for any $j \in \{1, \dots, n\}$.
    \item Showing that $c_{n+j} \geq -c_k$ for $j, k \in \{1, \dots, n\}$ such that:
    \begin{enumerate}
        \item $j = k$,
        \item $j = k + 1$ (or similarly, $k = j + 1$), and
        \item $j \geq k + 2$ (or similarly, $k \geq j + 2$).
    \end{enumerate}
\end{enumerate}

\subsubsection*{Case 1: Comparing the $j$th State Variable Constraints to the Input Constraints}

In this case, we will show that $c_{n+j} \geq \ubar{u}$ for any $j \in \{1, \dots, n\}$. The process of showing $\bar{u} \geq -c_j$ is identical, and thus omitted for conciseness. From \eqref{eqn:full_c_lower}, whenever $\vec{x} \in \bar{S}$, we have
\begin{align}
    c_{n+j} &\geq \bar{\gamma}_{j(n-j+1)}\sqrt{\bar{h}_{j(n-j+1)}} - \frac{\bar{\gamma}_{j(n-j)}^2}{2} - \bar{\epsilon}_{j(n-j+1)} \nonumber \\
    &\geq -\frac{\bar{\gamma}_{j(n-j)}^2}{2} \nonumber \\
    &= -\frac{\gamma_{n-1}^2}{2}
\end{align}
using the definition of $\bar{S}_{j(n-j+1)}$ in \eqref{eqn:full_safe_subsets_upper} for the second inequality and \eqref{eqn:gamma_epsilon_matching} for the final equation. Now, if $\gamma_{n-1}$ is chosen according to Algorithm \ref{alg:parameter_tuning}, Lemma \ref{lem:tuning} guarantees that \eqref{eqn:gamma_n-1_upper_limit} is satisfied. It follows from \eqref{eqn:gamma_n-1_upper_limit} that
\begin{equation}
    \frac{\gamma_{n-1}^2}{2} \leq (1 - \eta_{n-1})(-\ubar{u}) \leq -\ubar{u}
\end{equation}
since $\eta_{n-1} \in (0, 1)$ and $\ubar{u} < 0$. Therefore, $c_{n+j} \geq \ubar{u}$.

\subsubsection*{Case 2: Comparing the $j$th State Variable Constraints to the $k$th State Variable Constraint}

In this case, we will show that $c_{n+j} \geq -c_k$, or equivalently, that $c_{n+j} + c_k \geq 0$, for any pair $j, k \in \{1, \dots, n\}$ such that $j \geq k$. The case where $j < k$ is nearly identical and thus will be omitted for conciseness, but it is straightforward to verify using the same approach. From \eqref{eqn:full_c_lower} and \eqref{eqn:full_c_upper} and using \eqref{eqn:gamma_epsilon_matching}, whenever $\vec{x} \in \bar{S}$, we have
\begin{align}
    c_{n+j} + c_k &\geq \gamma_n\sqrt{\bar{h}_{j(n-j+1)}} - \frac{\bar{\gamma}_{j(n-j)}^2}{2} - \epsilon_n + \gamma_n\sqrt{\ubar{h}_{k(n-k+1)}} - \frac{\ubar{\gamma}_{k(n-k)}^2}{2} - \epsilon_n \nonumber \\
    &= \gamma_n(\sqrt{\bar{h}_{j(n-j+1)}} + \sqrt{\ubar{h}_{k(n-k+1)}}) - \frac{\bar{\gamma}_{j(n-j)}^2 + \ubar{\gamma}_{k(n-k)}^2}{2} - 2\epsilon_n \nonumber \\
    &\geq \gamma_n\sqrt{\bar{h}_{j(n-j+1)} + \ubar{h}_{k(n-k+1)}} - \frac{\bar{\gamma}_{j(n-j)}^2 + \ubar{\gamma}_{k(n-k)}^2}{2} - 2\epsilon_n.
\end{align}
Assuming for now that $j, k \leq n - 2$ (the alternative is implicitly handled later in the base cases of the proceeding recursion), we can then write
\begin{align}
    c_{n+j} + c_k &\geq \gamma_n\sqrt{\bar{h}_{j(n-j+1)} + \ubar{h}_{k(n-k+1)}} - \gamma_{n-1}^2 - 2\epsilon_n \nonumber \\
    &= \gamma_n\sqrt{-x_n + \gamma_{n-1}\sqrt{\bar{h}_{j(n-j)}} - \frac{\gamma_{n-2}^2}{2} - \epsilon_{n-1} + x_n + \gamma_{n-1}\sqrt{\ubar{h}_{k(n-k)}} - \frac{\gamma_{n-2}^2}{2} - \epsilon_{n-1}} - \gamma_{n-1}^2 - 2\epsilon_n \nonumber \\
    &= \gamma_n\sqrt{\gamma_{n-1}(\sqrt{\bar{h}_{j(n-j)}} + \sqrt{\ubar{h}_{k(n-k)}}) - \gamma_{n-2}^2 - 2\epsilon_{n-1}} - \gamma_{n-1}^2 - 2\epsilon_n \nonumber \\
    &\geq \gamma_n\sqrt{\gamma_{n-1}\sqrt{\bar{h}_{j(n-j)} + \ubar{h}_{k(n-k)}} - \gamma_{n-2}^2 - 2\epsilon_{n-1}} - \gamma_{n-1}^2 - 2\epsilon_n.
\end{align}
In this way, defining $C_{jk}^{n-j+1} := \bar{h}_{j(n-j+1)} + \ubar{h}_{k(n-k+1)}$, one can show the following recursion:
\begin{align}
    c_{n+j} + c_k &\geq \gamma_n\sqrt{C_{jk}^{n-j+1}} - \gamma_{n-1}^2 - 2\epsilon_n, \\
    C_{jk}^{n-j+1} &\geq \gamma_{n-1}\sqrt{C_{jk}^{n-j}} - \gamma_{n-2}^2 - 2\epsilon_{n-1}, \\
    &\;\; \vdots \nonumber
\end{align}
As suggested by the superscripts of $C_{jk}$, the base case of the recursion occurs when we reach $C_{jk}^1$, which is a function of $\ubar{h}_{j1}$. The specific form of the base case depends on the relative values of $j$ and $k$.

\subsubsection*{Case 2(a): $j = k$}

In this case, taking into account the fact that $\ubar{\gamma}_{j0} = \bar{\gamma}_{j0} = 0$ by definition, the base case is as follows:
\begin{align}
    C_{jk}^3 &= \bar{h}_{j3} + \ubar{h}_{j3} = -x_{j+2} + \bar{\gamma}_{j2}\sqrt{\bar{h}_{j2}} - \frac{\bar{\gamma}_{j1}^2}{2} - \bar{\epsilon}_{j2} + x_{j+2} + \ubar{\gamma}_{j2}\sqrt{\ubar{h}_{j2}} - \frac{\ubar{\gamma}_{j1}^2}{2} - \ubar{\epsilon}_{j2} \nonumber \\
    &= \gamma_{j+1}\sqrt{\bar{h}_{j2}} - \frac{\gamma_j^2}{2} - \epsilon_{j+1} + \gamma_{j+1}\sqrt{\ubar{h}_{j2}} - \frac{\gamma_j^2}{2} - \epsilon_{j+1} \nonumber \\
    &\geq \gamma_{j+1}\sqrt{\bar{h}_{j2} + \ubar{h}_{j2}} - \gamma_j^2 - 2\epsilon_{j+1} = \gamma_{j+1}\sqrt{C_{jk}^2} - \gamma_j^2 - 2\epsilon_{j+1}, \\
    C_{jk}^2 &= \bar{h}_{j2} + \ubar{h}_{j2} = -x_{j+1} + \bar{\gamma}_{j1}\sqrt{\bar{h}_{j1}} - \frac{\bar{\gamma}_{j0}^2}{2} - \bar{\epsilon}_{j1} + x_{j+1} + \ubar{\gamma}_{j1}\sqrt{\ubar{h}_{j1}} - \frac{\ubar{\gamma}_{j0}^2}{2} - \ubar{\epsilon}_{j1} \nonumber \\
    &= \gamma_j\sqrt{\bar{h}_{j1}} - \epsilon_j + \gamma_j\sqrt{\ubar{h}_{j1}} - \epsilon_j \nonumber \\
    &\geq \gamma_j\sqrt{\bar{h}_{j1} + \ubar{h}_{j1}} - 2\epsilon_j = \gamma_j\sqrt{C_{jk}^1} - 2\epsilon_j, \\
    C_{jk}^1 &= \bar{h}_{j1} + \ubar{h}_{j1} = \bar{x}_j - x_j + x_j - \ubar{x}_j = \bar{x}_j - \ubar{x}_j.
\end{align}
Now, from Lemma \ref{lem:tuning}, we know that \eqref{eqn:gamma_1_lower_limit}-\eqref{eqn:gamma_n-1_upper_limit} are satisfied. Using these inequalities, it is straightforward to show that:
\begin{align}
    C_{jk}^1 &= \bar{x}_j - \ubar{x}_j > 0, \\
    C_{jk}^2 &\geq \beta_j\gamma_j\sqrt{\bar{x}_j - \ubar{x}_j} > 0, \\
    C_{jk}^3 &\geq \alpha_{j+1}\beta_{j+1}\gamma_j^2 > 0, \\
    &\;\; \vdots \nonumber \\
    C_{jk}^{n-j+1} &\geq \alpha_{n-1}\beta_{n-1}\gamma_{n-2}^2 > 0, \\
    c_{n+j} + c_k &\geq \alpha_n\beta_n\gamma_{n-1}^2 > 0.
\end{align}

\subsubsection*{Case 2(b): $j = k + 1$}

In this case, taking into account the fact that $\ubar{\gamma}_{j0} = \bar{\gamma}_{j0} = 0$ by definition, the base case is as follows:
\begin{align}
    C_{jk}^3 &= \bar{h}_{j3} + \ubar{h}_{(j-1)4} = -x_{j+2} + \bar{\gamma}_{j2}\sqrt{\bar{h}_{j2}} - \frac{\bar{\gamma}_{j1}^2}{2} - \bar{\epsilon}_{j2} + x_{j+2} + \ubar{\gamma}_{(j-1)3}\sqrt{\ubar{h}_{(j-1)3}} - \frac{\ubar{\gamma}_{(j-1)2}^2}{2} - \ubar{\epsilon}_{(j-1)3} \nonumber \\
    &= \gamma_{j+1}\sqrt{\bar{h}_{j2}} - \frac{\gamma_j^2}{2} - \epsilon_{j+1} + \gamma_{j+1}\sqrt{\ubar{h}_{(j-1)3}} - \frac{\gamma_j^2}{2} - \epsilon_{j+1} \nonumber \\
    &\geq \gamma_{j+1}\sqrt{\bar{h}_{j2} + \ubar{h}_{(j-1)3}} - \gamma_j^2 - 2\epsilon_{j+1} = \gamma_{j+1}\sqrt{C_{jk}^2} - \gamma_j^2 - 2\epsilon_{j+1}, \\
    C_{jk}^2 &= \bar{h}_{j2} + \ubar{h}_{(j-1)3} = -x_{j+1} + \bar{\gamma}_{j1}\sqrt{\bar{h}_{j1}} - \frac{\bar{\gamma}_{j0}^2}{2} - \bar{\epsilon}_{j1} + x_{j+1} + \ubar{\gamma}_{(j-1)2}\sqrt{\ubar{h}_{(j-1)2}} - \frac{\ubar{\gamma}_{(j-1)1}^2}{2} - \ubar{\epsilon}_{(j-1)2} \nonumber \\
    &= \gamma_j\sqrt{\bar{h}_{j1}} - \epsilon_j + \gamma_j\sqrt{\ubar{h}_{(j-1)2}} - \frac{\gamma_{j-1}^2}{2} - \epsilon_j \nonumber \\
    &\geq \gamma_j\sqrt{\bar{h}_{j1} + \ubar{h}_{(j-1)2}} - \frac{\gamma_{j-1}^2}{2} - 2\epsilon_j = \gamma_j\sqrt{C_{jk}^1} - \frac{\gamma_{j-1}^2}{2} - 2\epsilon_j, \\
    C_{jk}^1 &= \bar{h}_{j1} + \ubar{h}_{(j-1)2} = \bar{x}_j - x_j + x_j + \ubar{\gamma}_{(j-1)1}\sqrt{\ubar{h}_{(j-1)1}} - \frac{\ubar{\gamma}_{(j-1)0}^2}{2} - \ubar{\epsilon}_{(j-1)1} = \bar{x}_j + \gamma_{j-1}\sqrt{\ubar{h}_{(j-1)1}} - \epsilon_{j-1} \nonumber \\
    &\geq \bar{x}_j
\end{align}
where the final inequality uses the fact that $\ubar{h}_{(j-1)1} \geq \frac{\epsilon_{j-1}^2}{\gamma_{j-1}^2}$ whenever $\vec{x} \in \bar{S}$. Now, from Lemma \ref{lem:tuning}, we know that \eqref{eqn:gamma_1_lower_limit}-\eqref{eqn:gamma_n-1_upper_limit} are satisfied. Using these inequalities, it is straightforward to show that:
\begin{align}
    C_{jk}^1 &\geq \bar{x}_j > 0, \\
    C_{jk}^2 &\geq \alpha_j\beta_j\gamma_{j-1}^2 > 0, \\
    C_{jk}^3 &\geq \alpha_{j+1}\beta_{j+1}\gamma_j^2 > 0, \\
    &\;\; \vdots \nonumber \\
    &\;\; \vdots \nonumber \\
    C_{jk}^{n-j+1} &\geq \alpha_{n-1}\beta_{n-1}\gamma_{n-2}^2 > 0, \\
    c_{n+j} + c_k &\geq \alpha_n\beta_n\gamma_{n-1}^2 > 0.
\end{align}

\subsubsection*{Case 2(c): $j \geq k + 2$}

In this case, taking into account the fact that $\ubar{\gamma}_{j0} = \bar{\gamma}_{j0} = 0$ by definition, the base case is as follows:
\begin{align}
    C_{jk}^3 &= \bar{h}_{j3} + \ubar{h}_{k(j-k+3)} \nonumber \\
    &= -x_{j+2} + \bar{\gamma}_{j2}\sqrt{\bar{h}_{j2}} - \frac{\bar{\gamma}_{j1}^2}{2} - \bar{\epsilon}_{j2} + x_{j+2} + \ubar{\gamma}_{k(j-k+2)}\sqrt{\ubar{h}_{k(j-k+2)}} - \frac{\ubar{\gamma}_{k(j-k+1)}^2}{2} - \ubar{\epsilon}_{k(j-k+2)} \nonumber \\
    &= \gamma_{j+1}\sqrt{\bar{h}_{j2}} - \frac{\gamma_j^2}{2} - \epsilon_{j+1} + \gamma_{j+1}\sqrt{\ubar{h}_{k(j-k+2)}} - \frac{\gamma_j^2}{2} - \epsilon_{j+1} \nonumber \\
    &\geq \gamma_{j+1}\sqrt{\bar{h}_{j2} + \ubar{h}_{k(j-k+2)}} - \gamma_j^2 - 2\epsilon_{j+1} = \gamma_{j+1}\sqrt{C_{jk}^2} - \gamma_j^2 - 2\epsilon_{j+1}, \\
    C_{jk}^2 &= \bar{h}_{j2} + \ubar{h}_{k(j-k+2)} \nonumber \\
    &= -x_{j+1} + \bar{\gamma}_{j1}\sqrt{\bar{h}_{j1}} - \frac{\bar{\gamma}_{j0}^2}{2} - \bar{\epsilon}_{j1} + x_{j+1} + \ubar{\gamma}_{k(j-k+1)}\sqrt{\ubar{h}_{k(j-k+1)}} - \frac{\ubar{\gamma}_{k(j-k)}^2}{2} - \ubar{\epsilon}_{k(j-k+1)} \nonumber \\
    &= \gamma_j\sqrt{\bar{h}_{j1}} - \epsilon_j + \gamma_j\sqrt{\ubar{h}_{k(j-k+1)}} - \frac{\gamma_{j-1}^2}{2} - \epsilon_j \nonumber \\
    &\geq \gamma_j\sqrt{\bar{h}_{j1} + \ubar{h}_{k(j-k+1)}} - \frac{\gamma_{j-1}^2}{2} - 2\epsilon_j = \gamma_j\sqrt{C_{jk}^1} - \frac{\gamma_{j-1}^2}{2} - 2\epsilon_j, \\
    C_{jk}^1 &= \bar{h}_{j1} + \ubar{h}_{k(j-k+1)} = \bar{x}_j - x_j + x_j + \ubar{\gamma}_{k(j-k)}\sqrt{\ubar{h}_{k(j-k)}} - \frac{\ubar{\gamma}_{k(j-k-1)}^2}{2} - \ubar{\epsilon}_{k(j-k)} \nonumber \\
    &= \bar{x}_j + \gamma_{j-1}\sqrt{\ubar{h}_{k(j-k)}} - \frac{\gamma_{j-2}^2}{2} - \epsilon_{j-1} \nonumber \\
    &\geq \bar{x}_j - \frac{\gamma_{j-2}^2}{2}
\end{align}
where the final inequality uses the fact that $\ubar{h}_{k(j-k)} \geq \frac{\epsilon_{j-1}^2}{\gamma_{j-1}^2}$ whenever $\vec{x} \in \bar{S}$. Now, from Lemma \ref{lem:tuning}, we know that \eqref{eqn:gamma_1_lower_limit}-\eqref{eqn:gamma_n-1_upper_limit} are satisfied. Using these inequalities, it is straightforward to show that:
\begin{align}
    C_{jk}^1 &\geq \eta_{j-2}\bar{x}_j > 0, \\
    C_{jk}^2 &\geq \alpha_j\beta_j\gamma_{j-1}^2 > 0, \\
    C_{jk}^3 &\geq \alpha_{j+1}\beta_{j+1}\gamma_j^2 > 0, \\
    &\;\; \vdots \nonumber \\
    &\;\; \vdots \nonumber \\
    C_{jk}^{n-j+1} &\geq \alpha_{n-1}\beta_{n-1}\gamma_{n-2}^2 > 0, \\
    c_{n+j} + c_k &\geq \alpha_n\beta_n\gamma_{n-1}^2 > 0.
\end{align}

\subsubsection*{Step 3: Proving that $\bar{S}$ is Non-Empty}

We accomplish this final step by proving that $\vec{x}_* := [x_{1*}, 0, \dots, 0]^\top$ is in $\bar{S}$ for any $x_{1*} \in [\ubar{x}_1 + \delta, \bar{x}_1 - \delta]$. We will only prove here that $\vec{x}_* \in \bigcap_{j=1}^n\bigcap_{i=1}^{n-j+1} \ubar{S}_{ji}$ for $\ubar{S}_{ji}$ as in \eqref{eqn:full_safe_subsets_lower}, as the proof that $\vec{x}_* \in \bigcap_{j=1}^n\bigcap_{i=1}^{n-j+1} \bar{S}_{ji}$ is virtually identical and thus omitted for conciseness. Consider first the constraints on $x_1$. Using \eqref{eqn:lower_h_1}-\eqref{eqn:lower_h_i}, \eqref{eqn:gamma_epsilon_matching}, and \eqref{eqn:gamma_1_lower_limit}-\eqref{eqn:gamma_n-1_upper_limit}, we have
\begin{align}
    \ubar{h}_{11}(\vec{x}_*) &= x_{1*} - \ubar{x}_1 \geq \delta, \label{eqn:Sbar_non_empty_1} \\
    \ubar{h}_{12}(\vec{x}_*) &= 0 + \gamma_1\sqrt{\ubar{h}_{11}(\vec{x}_*)} - \epsilon_1 \geq \gamma_1\sqrt{\delta} - \epsilon_1 \geq \beta_1\gamma_1\sqrt{\delta}, \\
    \ubar{h}_{13}(\vec{x}_*) &= 0 + \gamma_2\sqrt{\ubar{h}_{12}(\vec{x}_*)} - \frac{\gamma_1^2}{2} - \epsilon_2 \geq \gamma_2\sqrt{\beta_1\gamma_1\sqrt{\delta}} - \frac{\gamma_1^2}{2} - \epsilon_2 \geq \alpha_2\beta_2\gamma_1^2, \\
    \ubar{h}_{14}(\vec{x}_*) &= 0 + \gamma_3\sqrt{\ubar{h}_{13}(\vec{x}_*)} - \frac{\gamma_2^2}{2} - \epsilon_3 \geq \gamma_3\sqrt{\alpha_2\beta_2}\gamma_1 - \frac{\gamma_2^2}{2} - \epsilon_3 \geq \alpha_3\beta_3\gamma_2^2, \\
    &\;\; \vdots \nonumber \\
    \ubar{h}_{1n}(\vec{x}_*) &= 0 + \gamma_{n-1}\sqrt{\ubar{h}_{1(n-1)}(\vec{x}_*)} - \frac{\gamma_{n-2}^2}{2} - \epsilon_{n-1} \geq \gamma_{n-1}\sqrt{\alpha_{n-2}\beta_{n-2}}\gamma_{n-3} - \frac{\gamma_{n-2}^2}{2} - \epsilon_{n-1} \nonumber \\
    &\geq \alpha_{n-1}\beta_{n-1}\gamma_{n-2}^2.
\end{align}
Now considering the constraints on $x_j$ for any $j \in \{2, \dots, n\}$, using \eqref{eqn:lower_h_1}-\eqref{eqn:lower_h_i}, \eqref{eqn:gamma_epsilon_matching}, and \eqref{eqn:gamma_1_lower_limit}-\eqref{eqn:gamma_n-1_upper_limit}, we have
\begin{align}
    \ubar{h}_{j1}(\vec{x}_*) &= 0 - \ubar{x}_j = -\ubar{x}_j, \\
    \ubar{h}_{j2}(\vec{x}_*) &= 0 + \gamma_j\sqrt{\ubar{h}_{j1}} - \epsilon_j = \gamma_j\sqrt{-\ubar{x}_j} - \epsilon_j \geq \beta_j\gamma_j\sqrt{-\ubar{x}_j} \\
    \ubar{h}_{j3}(\vec{x}_*) &= 0 + \gamma_{j+1}\sqrt{\ubar{h}_{j2}} - \frac{\gamma_j^2}{2} - \epsilon_{j+1} \geq \gamma_{j+1}\sqrt{\beta_j\gamma_j\sqrt{-\ubar{x}_j}} - \frac{\gamma_j^2}{2} - \epsilon_{j+1} \geq \alpha_{j+1}\beta_{j+1}\gamma_j^2 \\
    \ubar{h}_{j4}(\vec{x}_*) &= 0 + \gamma_{j+2}\sqrt{\ubar{h}_{j3}} - \frac{\gamma_{j+1}^2}{2} - \epsilon_{j+2} \geq \gamma_{j+2}\sqrt{\alpha_{j+1}\beta_{j+1}}\gamma_j - \frac{\gamma_{j+1}^2}{2} - \epsilon_{j+2} \geq \alpha_{j+2}\beta_{j+2}\gamma_{j+1}^2 \\
    &\;\; \vdots \nonumber \\
    \ubar{h}_{j(n-j+1)}(\vec{x}_*) &= 0 + \gamma_{n-1}\sqrt{\ubar{h}_{j(n-j)}} - \frac{\gamma_{n-2}^2}{2} - \epsilon_{n-1} \geq \gamma_{n-1}\sqrt{\alpha_{n-2}\beta_{n-2}}\gamma_{n-3} - \frac{\gamma_{n-2}^2}{2} - \epsilon_{n-1} \nonumber \\
    &\geq \alpha_{n-1}\beta_{n-1}\gamma_{n-2}^2. \label{eqn:Sbar_non_empty_2}
\end{align}

Now, $\vec{x}_* \in \bigcap_{j=1}^n\bigcap_{i=1}^{n-j+1} \ubar{S}_{ji}$ is equivalent to $\ubar{h}_{ji}(\vec{x}_*) \geq \frac{\epsilon_{i+j-1}^2}{\gamma_{i+j-1}^2}$ for all $j \in \{1, \dots, n\}$ and all $i \in \{1, \dots, n-j+1\}$, or equivalently, $\gamma_{i+j-1}\sqrt{\ubar{h}_{ji}(\vec{x}_*)} - \epsilon_{i+j-1} \geq 0$. As already shown in \eqref{eqn:Sbar_non_empty_1}-\eqref{eqn:Sbar_non_empty_2}, we have:
\begin{align}
    \gamma_1\sqrt{\ubar{h}_{11}(\vec{x}_*)} - \epsilon_1 &\geq \gamma_1\sqrt{\delta} - \epsilon_1 \geq \beta_1\gamma_1\sqrt{\delta} > 0, \\
    \gamma_2\sqrt{\ubar{h}_{12}(\vec{x}_*)} - \epsilon_2 &\geq \gamma_2\sqrt{\beta_1\gamma_1\sqrt{\delta}} - \epsilon_2 \geq \alpha_2\beta_2\gamma_1^2 + \frac{\gamma_1^2}{2} > 0, \\
    \gamma_i\sqrt{\ubar{h}_{1i}(\vec{x}_*)} - \epsilon_i &\geq \gamma_i\sqrt{\alpha_{i-1}\beta_{i-1}}\gamma_{i-2} - \epsilon_i \geq \alpha_i\beta_i\gamma_{i-1}^2 + \frac{\gamma_{i-1}^2}{2} > 0\ \forall i \in \{3, \dots, n-1\}, \\
    \gamma_j\sqrt{\ubar{h}_{j1}(\vec{x}_*)} - \epsilon_j &\geq \gamma_j\sqrt{-\ubar{x}_j} - \epsilon_j \geq \beta_j\gamma_j\sqrt{-\ubar{x}_j} > 0\ \forall j \in \{2, \dots, n\}, \\
    \gamma_{j+1}\sqrt{\ubar{h}_{j2}(\vec{x}_*)} - \epsilon_{j+1} &\geq \gamma_{j+1}\sqrt{\beta_j\gamma_j\sqrt{-\ubar{x}_j}} - \epsilon_{j+1} \geq \alpha_{j+1}\beta_{j+1}\gamma_j^2 + \frac{\gamma_j^2}{2} > 0\ \forall j \in \{2, \dots, n\}, \\
    \gamma_{i+j-1}\sqrt{\ubar{h}_{ji}(\vec{x}_*)} - \epsilon_{i+j-1} &\geq \gamma_{i+j-1}\sqrt{\alpha_{i+j-2}\beta_{i+j-2}}\gamma_{i+j-3} - \epsilon_{i+j-1} \geq \alpha_{i+j-1}\beta_{i+j-1}\gamma_{i+j-2}^2 + \frac{\gamma_{i+j-2}^2}{2} > 0 \nonumber \\
    &\indenti{\geq}\ \forall i \in \{3, \dots, n-j\},
\end{align}
Finally, to complete the proof, one can show using \eqref{eqn:gamma_1_lower_limit}-\eqref{eqn:gamma_n-1_upper_limit} in the same manner as above that:
\begin{align}
    \gamma_n\sqrt{\ubar{h}_{1n}(\vec{x}_*)} - \epsilon_n &\geq \gamma_n\sqrt{\alpha_{n-1}\beta_{n-1}}\gamma_{n-2} - \epsilon_n \geq \alpha_n\beta_n\gamma_{n-1}^2, \\
    \gamma_n\sqrt{\ubar{h}_{j(n-j+1)}(\vec{x}_*)} - \epsilon_n &\geq \gamma_n\sqrt{\alpha_{n-1}\beta_{n-1}}\gamma_{n-2} - \epsilon_n \geq \alpha_n\beta_n\gamma_{n-1}^2,
\end{align}
thus completing the proof. \hfill $\Box$
\section{Additional Simulation Results} \label{app:simulations}

In this appendix, we provide additional supporting simulation results demonstrating the efficacy of our approach. Figures \ref{fig:traj_2_top_view}-\ref{fig:traj_3_thrusts} show two cases where the reference trajectory leads to target locations outside of the safe set, starting from different initial conditions. In both cases, the plant simultaneously obeys input constraints and maneuvers itself to be as close to the reference trajectory as possible, ending as close to the target as possible, while remaining inside the safe set. Finally, Figures \ref{fig:traj_3_small_top_view}-\ref{fig:traj_3_small_thrusts} are a repeat of the simulation scenario in Figures \ref{fig:traj_3_top_view}-\ref{fig:traj_3_thrusts}, but with input limitations that are artificially made severely restrictive. As before, the plant obeys the input restrictions and keeps itself inside the safe set.

\begin{figure}[h]
    \centering
    \begin{subfigure}[t]{0.48\textwidth}
        \centering
        \includegraphics[width=\textwidth]{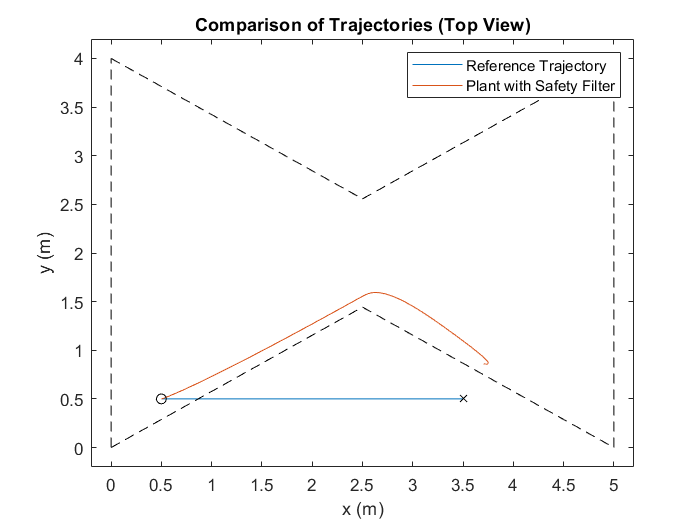}
        \caption{Top-down view of the plant's trajectory under the safety filter. The plant stays inside the safe set (dashed black lines) even though the target is unsafe.}
        \label{fig:traj_2_top_view}
    \end{subfigure}
    \hfill
    % \begin{subfigure}[t]{0.48\textwidth}
    %     \centering
    %     \includegraphics[width=\textwidth]{Images/x_y_pos_large_input_traj_2.png}
    %     \caption{$x(t)$ and $y(t)$ of the plant under the safety filter.}
    %     \label{fig:traj_2_x_y_pos}
    % \end{subfigure}
    % \hfill
    \begin{subfigure}[t]{0.48\textwidth}
        \centering
        \includegraphics[width=\textwidth]{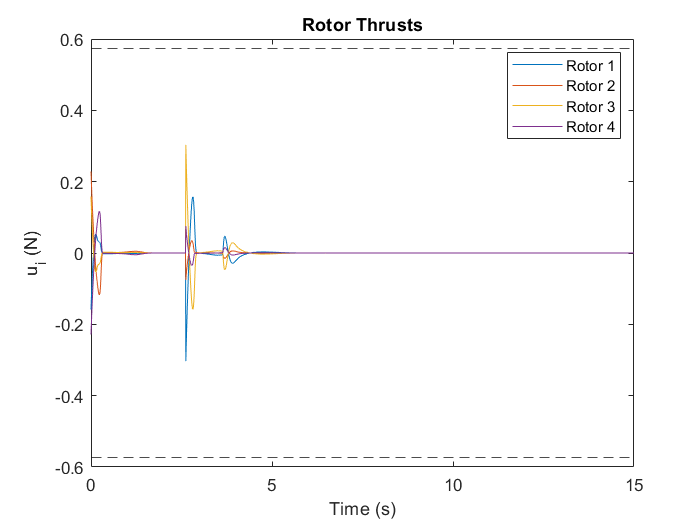}
        \caption{Individual rotor thrusts resulting from the safety filter. All rotors obey the input limits (dashed black lines).}
        \label{fig:traj_2_thrusts}
    \end{subfigure}
\end{figure}

\begin{figure}[h]
    \centering
    \begin{subfigure}[t]{0.48\textwidth}
        \centering
        \includegraphics[width=\textwidth]{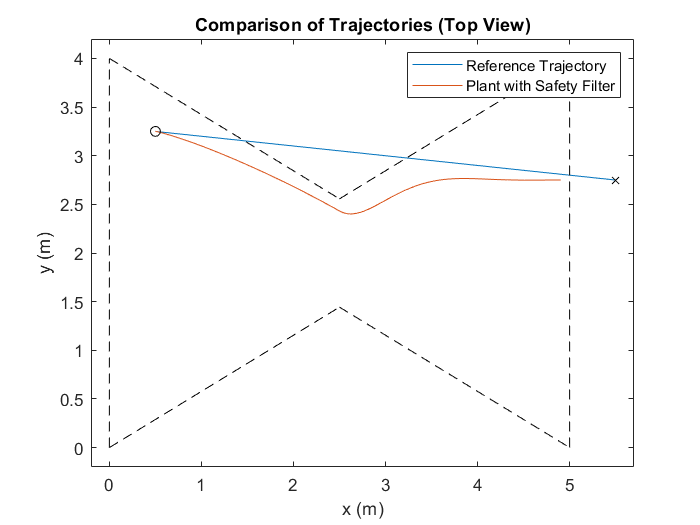}
        \caption{Top-down view of the plant's trajectory under the safety filter. The plant stays inside the safe set (dashed black lines) even though the target is unsafe.}
        \label{fig:traj_3_top_view}
    \end{subfigure}
    \hfill
    % \begin{subfigure}[t]{0.48\textwidth}
    %     \centering
    %     \includegraphics[width=\textwidth]{Images/x_y_pos_large_input_traj_3.png}
    %     \caption{$x(t)$ and $y(t)$ of the plant under the safety filter.}
    %     \label{fig:traj_3_x_y_pos}
    % \end{subfigure}
    % \hfill
    \begin{subfigure}[t]{0.48\textwidth}
        \centering
        \includegraphics[width=\textwidth]{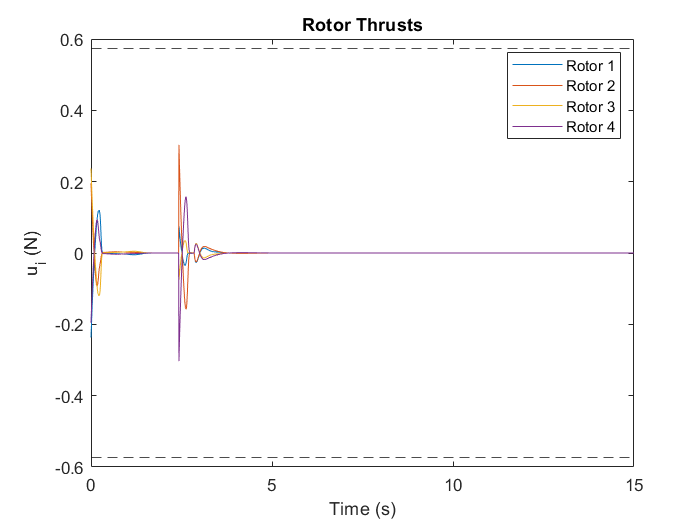}
        \caption{Individual rotor thrusts resulting from the safety filter. All rotors obey the input limits (dashed black lines).}
        \label{fig:traj_3_thrusts}
    \end{subfigure}
\end{figure}

\begin{figure}[h]
    \centering
    \begin{subfigure}[t]{0.48\textwidth}
        \centering
        \includegraphics[width=\textwidth]{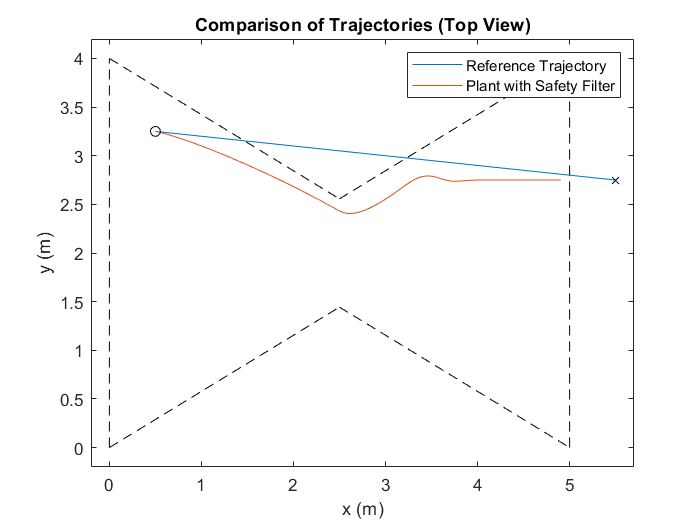}
        \caption{Top-down view of the plant's trajectory under the safety filter. The plant stays inside the safe set (dashed black lines) even though the target is unsafe.}
        \label{fig:traj_3_small_top_view}
    \end{subfigure}
    \hfill
    % \begin{subfigure}[t]{0.48\textwidth}
    %     \centering
    %     \includegraphics[width=\textwidth]{Images/x_y_pos_small_input_traj_3.png}
    %     \caption{$x(t)$ and $y(t)$ of the plant under the safety filter.}
    %     \label{fig:traj_3_small_x_y_pos}
    % \end{subfigure}
    % \hfill
    \begin{subfigure}[t]{0.48\textwidth}
        \centering
        \includegraphics[width=\textwidth]{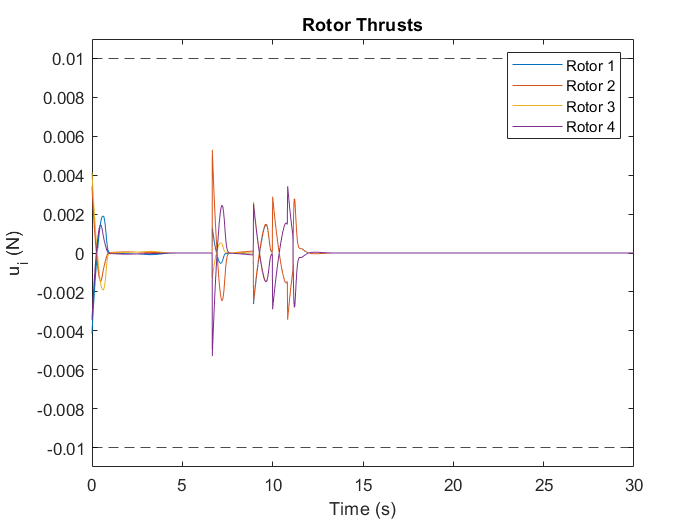}
        \caption{Individual rotor thrusts resulting from the safety filter. All rotors obey the (artificially severely restrictive) input limits (dashed black lines).}
        \label{fig:traj_3_small_thrusts}
    \end{subfigure}
\end{figure}

\end{document}